\documentclass[11pt]{amsart}

\usepackage{graphicx}
\usepackage{tikz}
\usetikzlibrary{matrix}
\usepackage{amssymb,amsthm}
\usepackage{graphbox}
\usepackage{upgreek}
\usepackage{tcolorbox}
\usepackage{tabto}
\usepackage{braket}
\usepackage[export]{adjustbox}

\usepackage[margin=1in]{geometry}
\usepackage[pagebackref, colorlinks = true, linkcolor = blue, urlcolor  = blue, citecolor = red]{hyperref}
\usepackage{xcolor}
\definecolor{darkgreen}{rgb}{0.09, 0.45, 0.27}
\definecolor{darkred}{rgb}{0.55, 0.0, 0.0}

\renewcommand{\epsilon}{\varepsilon}
\renewcommand{\phi}{\varphi}
\newcommand{\C}{\mathbb{C}}
\newcommand{\bC}{\mathbb{C}}

\DeclareMathOperator{\Tr}{Tr}
\DeclareMathOperator{\id}{id}

\newcommand{\E}{\mathbb{E}}
\renewcommand{\P}{\mathbb{P}}

\newtheorem{theorem}{Theorem}[section]
\newtheorem{definition}[theorem]{Definition}
\newtheorem{proposition}[theorem]{Proposition}
\newtheorem{corollary}[theorem]{Corollary}
\newtheorem{lemma}[theorem]{Lemma}
\newtheorem{remark}[theorem]{Remark}

\newcommand{\Sym}{\mathrm{Sym}}

\newcommand{\Pcomp}{P_{\overline{\Omega}}}
\newcommand{\ketbra}[2]{| #1 \rangle \langle #2 |}

\begin{document}

\title[Entanglement criteria for the bosonic and fermionic induced ensembles]{Entanglement criteria for the\\bosonic and fermionic induced ensembles}

\author{{S. Dartois}}\email{stephane.dartois@u-bordeaux.fr}
\address{Univ. Bordeaux, LaBRI, UMR 5800, F-33400 Talence, France }

\author{Ion Nechita}
\email{ion.nechita@univ-tlse3.fr}
\address{Laboratoire de Physique Th\'eorique, Universit\'e de Toulouse, CNRS, UPS, France.}

\author{{A. Tanasa}}\email{ntanasa@u-bordeaux.fr}
\address{Univ. Bordeaux, LaBRI, UMR 5800, F-33400 Talence, France and
DFT, H. Hulubei Nat. Inst. Nucl. Engineering, P.O.Box MG-6, 077125 Magurele, Romania}

\begin{abstract}
We introduce the bosonic and fermionic ensembles of density matrices and study their entanglement. In the fermionic case, we show that random bipartite fermionic density matrices have non-positive partial transposition, hence they are typically entangled. The similar analysis in the bosonic case is more delicate, due to a large positive outlier eigenvalue. We compute the asymptotic ratio between the size of the environment and the size of the system Hilbert space for which random bipartite bosonic density matrices fail the PPT criterion, being thus entangled. We also relate moment computations for tensor-symmetric random matrices to evaluations of the circuit-counting and interlace graph polynomials for directed graphs.
\end{abstract}

\date{\today}

\maketitle

\vspace{1cm}

{\bf Keywords: } random fermionic and bosonic states, random tensors and matrices, entanglement, circuit polynomial 

\vspace{1cm}

\tableofcontents

\newpage

\section{Introduction}

In Quantum Mechanics, quantum states are modelled by vectors in a complex Hilbert space. In the presence of an environment, one uses the formalism of density matrices to describe quantum states. The dimension of the corresponding Hilbert space grows exponentially with the number of quantum particles one wants to describe, hence it is very difficult to precisely characterize the interesting physical properties of arbitrary density matrices. One resorts then to understanding the properties of \emph{typical quantum states}, in the sense that one would like to find properties which hold with large probability, for some suitable (natural from a physical perspective) probability measure on the set of density matrices. 

Several such natural probability measures have been considered in the literature: the induced measure (closely related to the Wishart ensemble from random matrix theory) \cite{zyczkowski2001induced}, the Bures measure (motivated by statistical considerations) \cite{hall1998random}, random matrix product states (motivated by condensed matter theory) \cite{garnerone2010typicality}, and random graph states (which encode a given entanglement structure) \cite{collins2010randoma}. Our first contribution in this work is to introduce density matrix ensembles for bipartite quantum systems which have a given \emph{symmetry}. We consider respectively, the symmetric and the anti-symmetric subspace of a tensor product of Hilbert spaces, and introduce random density matrices supported on these subspaces. Since the construction is very similar to that of the induced ensemble \cite{zyczkowski2001induced}, we call these the \emph{bosonic}, respectively the \emph{fermionic induced ensembles} of mixed quantum states. 

We then study the entanglement properties of states from these ensembles, focusing on the bipartite case (two fermions and two bosons). We analyze the spectrum of the \emph{partial transposition} \cite{horodecki1996separability,peres1996separability} of these random density matrices in the large $N$ limit. For the bipartite fermionic ensemble, we find (see Corollary \ref{cor:fermionic-threshold}) that a typical fermionic mixed state is entangled. This is due to the presence of a large negative eigenvalue in the spectrum of the partial transposition of these states. 
In the bosonic case, the situation is more complicated, since the symmetry of the state is responsible for a large positive eigenvalue of the partial transposition. This outlier eigenvalue makes the study of the spectrum of the partial transposition more delicate. The asymptotic spectrum of the partial transposition of a random bosonic density matrix is computed in Theorem \ref{thm:main}, which we state informally here: 

\begin{theorem}
Let $\rho_N$ be a $N^2 \times N^2$ element from the bosonic ensemble. The spectrum of $\rho_N$ contains, in the large $N$ limit: 
\begin{itemize}
    \item a large, \emph{outlier}, positive eigenvalue, of order $1/N$
    \item a \emph{bulk}, containing $N^2-1$ eigenvalues, of order $1/N^2$, whose empirical distribution follows, when properly normalized, a \emph{shifted semicircle distribution}. 
\end{itemize}
The bulk of the spectrum is characterized by the shape parameter of the bosonic induced ensemble, and we observe a \emph{threshold phenomenon}: the limiting probability distribution of the bulk eigenvalues has negative support if and only if the shape parameter $c$ is smaller than 4. 
\end{theorem}

We discuss the implications of this result for the entanglement of typical bosonic states in Corollary \ref{cor:bosonic-threshold}: when the shape parameter of the bosonic induced ensemble is smaller than $4$, a typical symmetric mixed state will be entangled. We also show that the \emph{realignment criterion} \cite{chen2002matrix, rudolph2003cross} gives exactly the same entanglement threshold as the positive partial transposition (PPT) criterion. 

We analyze the asymptotic spectrum of the random matrices we consider (and their partial transposition) using the moment method. Due to the imposed symmetry of the models, direct computations are cumbersome; we put forward a novel connection between such moment computations and a well-known graph polynomial: the \emph{circuit counting polynomial} of directed graphs. In combinatorics and graph theory, this polynomial has received a lot of attention recently, in relation to the \emph{interlace polynomial} \cite{arratia2004interlace}. We use this connection in our derivations to bound the asymptotic moments of the random matrices we are considering, showcasing the importance of the newly discovered parallel between moment computations and graph polynomials. 

The paper is organized as follows. In Section \ref{sec:first-def} we recall some basic definitions from quantum information theory and from the theory of symmetric operators. In Sections \ref{sec:fermionic} and \ref{sec:bosonic} we introduce, respectively, the fermionic and the bosonic induced ensembles. In Section \ref{sec:main-result} we state the main result on
the asymptotic spectrum of a random bosonic quantum state. In Section \ref{sec:graph-poly} we recall the main facts on the circuit counting graph polynomial, relating it to generalized traces of the symmetric projection. This connection is used in Section \ref{sec:moments} to express the moments of the partial transposition of a random bosonic density matrix to evaluations of the circuit counting graph polynomial. Sections \ref{sec:t-channel}, \ref{sec:diagrammatics},  \ref{sec:bounds}, \ref{sec:tensor-eval} contain technical results used in Section \ref{sec:proof-main-result} for the proof of the main result. We conclude in Section \ref{sec:conclusion}, where we also lay out some directions for future research. 

\section{Quantum states, entanglement, symmetry}\label{sec:first-def}

In this section we recall the basic notions from quantum information theory which motivate our work, and then gather some basic mathematical facts regarding symmetric operators needed later. We start with a brief overview of quantum states, the notion of quantum entanglement, and the partial transposition entanglement criterion.

In quantum mechanics, quantum states are modelled by \emph{density matrices}. Each quantum system comes with a complex Hilbert space $\mathcal H$, which we shall consider here to be finite dimensional: $\mathcal H \approx \mathbb C^N$; we say that the quantum system has $N$ degrees of freedom. A density matrix is a positive semidefinite operator acting on $\mathcal H$, with unit trace:
$$\rho \in \mathcal M_N(\mathbb C), \, \rho \geq 0, \, \Tr \rho = 1.$$

The Hilbert space corresponding to two quantum systems is obtained by taking the \emph{tensor product} of the individual Hilbert spaces: 
$$\mathcal H_{AB} = \mathcal H_A \otimes \mathcal H_B.$$
A bipartite quantum state $\rho_{AB}$ is said to be \emph{separable} if it can be written as a convex mixture of product states: 
$$\rho_{AB} = \sum_{i=1}^k p_i \alpha_i \otimes \beta_i,$$
for density matrices $\alpha_i$, $\beta_i$, and probabilities $p_i$. Non-separable quantum states are called \emph{entangled}. Deciding whether a given bipartite density matrix is separable or entangled is a computationally hard problem \cite{gurvits2004classical}. For this reason, several computationally efficient necessary or sufficient conditions for entanglement have been developed \cite[Section VI.B]{horodecki2009quantum}.

In this paper, we shall be concerned with the most famous entanglement criterion, the \emph{positive partial transpose} criterion \cite{peres1996separability,horodecki1996separability} (we shall also briefly discuss the realignment criterion in Section \ref{sec:bosonic}). The starting observation is that the partial transposition of a separable state is positive semidefinite: 

$$\rho_{AB}^\Gamma := [\operatorname{id} \otimes \Tr](\rho_{AB}) = \sum_i \alpha_i \otimes \beta_i^\top \geq 0.$$

Hence, if the partial transposition of a quantum state $\rho_{AB}$ is \emph{not} positive semidefinite, the state must be entangled. This criterion is known to be exact only if the total Hilbert space dimension is at most 6; we refer the reader to \cite[Section VI.B.1]{horodecki2009quantum} for further properties.

\bigskip

We shall now discuss several mathematical considerations regarding symmetry in quantum information theory; more precisely, we shall investigate the symmetric subspace of a tensor product of Hilbert spaces, and the related operators.  Consider the $r^{\textrm{th}}$ tensor power $\mathcal{H}^{\otimes r}$ of $\mathcal H$. A natural basis of $\mathcal{H}^{\otimes r}$ is given by the set $\{e_{i_1}\otimes \ldots \otimes e_{i_r} \, : \,  i_1,\ldots, i_r \in [N]\}$ where the $e_i$'s form a basis of $\mathcal{H}$, and $[N]$ denotes the set $\{1,2, \ldots, N\}$. Let $\sigma\in S_r$. We define the action of $\sigma$ on $\mathcal{H}^{\otimes r}$ by extending linearly its action on basis elements, so that,
\begin{equation}
    \sigma(e_{i_1}\otimes \ldots \otimes e_{i_r})=e_{i_{\sigma(1)}}\otimes \ldots \otimes e_{i_{\sigma(r)}}.
\end{equation}
We now define two orthogonal projectors 
\begin{equation}
   P_{s,r}=\frac 1 {r!}\sum_{\sigma\in S_r} \sigma \quad \textrm{ and } \quad P_{a,r}=\frac1{r!}\sum_{\sigma\in S_r} \epsilon(\sigma) \sigma
\end{equation}
where $\epsilon(\sigma)$ is the signature of the permutation $\sigma$. $P_{s,r}$ are projectors on the symmetric subspace $\mathrm{Sym}_r \mathcal{H}$ and $P_{a,r}$ are projectors on the antisymmetric subspace (or $r^{\textrm{th}}$ exterior power) $\wedge^r\mathcal{H}$. 
We define the $r^{\textrm{th}}$ symmetric power (resp. the $r^{\textrm{th}}$ exterior power) of $\mathcal H$ as the subspace $P_{s,r}(\mathcal{H}^{\otimes r})$ (resp. $P_{a,r}(\mathcal{H}^{\otimes r})$). We recall
\begin{equation}
    N_s[r]:=\mathrm{dim }\ \mathrm{Sym}_r \mathcal{H}= \binom{N+r-1}{r} \ \mathrm{ and } \ N_a[r]:= \mathrm{dim }\ \wedge^r\mathcal{H}=\binom{N}{r}.
\end{equation}
The orthogonal projection on the symmetric subspace can be written as \cite{harrow2013church}
\begin{equation}\label{eq:def-P-sym}
P_{s,r} =  \binom{N+r-1}{r} \int_{\|x\|=1} (xx^*)^{\otimes r} \mathrm{d}x,
\end{equation}
For example, in the bipartite case, we have
$$P_{s,2} = \frac 1 2 (I+F) \ \textrm{and } P_{a,2}=\frac12(I-F)$$
where $F$ is the \emph{flip} (or \emph{swap}) operator $F x \otimes y = y \otimes x$ while $I$ is the identity permutation and is represented by a $N^2\times N^2$ identity matrix. Bipartite tensors can be identified with matrices in a straightforward way, and, in that case, we have $P_{s,2} A = (A+A^\top)/2$ and $P_{a,2} A=(A-A^\top)/2$. Note that this corresponds to symmetrizing and anti-symmetrizing matrices; in the complex case, we do not obtain however self-adjoint or skew-adjoint matrix, since we are not taking complex conjugates of the entries. Since most of our work will use these projectors for $r=2$, we use a simplified notation for this case and denote $P_s=P_{s,2}$ and $P_a=P_{a,2}$.

We state below a computation which will be useful later (see \cite{holmes2018partial} for a related problem, stated in representation theoretic language).
\begin{lemma}\label{lem:partial-trace-P-sym}
The partial trace of the projection on the symmetric subspace can be computed as follows: (below, $0 \leq r \leq k$): 
$$[\id_r \otimes \Tr_{k-r}]P_{s, k} = \frac{N[k]}{N[r]} P_{s,r}.$$
\end{lemma}
\begin{proof}
We use the integral representation from eq.~\eqref{eq:def-P-sym}:
\begin{align*}[\id_r \otimes \Tr_{k-r}]P_{s, k} &= N[k] \int_{\|x\|=1} [\id_r \otimes \Tr_{k-r}]\left((xx^*)^{\otimes k}\right) \mathrm{d}x,\\
&= N[k] \int_{\|x\|=1} (xx^*)^{\otimes r} \mathrm{d}x\\
&= \frac{N[k]}{N[r]} P_{s,r}.
\end{align*}
\end{proof}

\section{The fermionic induced ensemble}\label{sec:fermionic}

In this section we introduce the first random matrix model for quantum states that we shall study in this paper, that is the matrix model for random fermionic states. Random quantum states have played an important role in quantum physics and quantum information theory, being used to study properties of typical quantum systems and also as a source of examples of states with interesting properties. 

Ensembles of random density matrices have been considered by several authors, having different physical and mathematical motivations \cite{zyczkowski2001induced,zyczkowski2003hilbert,nechita2007asymptotics,garnerone2010typicality,zyczkowski2011generating}. The obvious candidate, the \emph{Lebesgue measure} on the convex body of density matrices, is part of a one-parameter family of probability measures, called the \emph{induced ensembles} \cite{zyczkowski2001induced}. The general setting is as follows. Consider a random pure state $\ket \psi \in \mathcal H \otimes \mathcal H_E$, uniformly distributed on the unit sphere the tensor product between the Hilbert space of interest $\mathcal H$ and an auxiliary Hilbert space $\mathcal H_E$, called the \emph{environment}. Define a random density matrix 
$$\rho := \Tr_E \ketbra{\psi}{\psi}$$
by tracing out the environment. In this way, we obtain a random matrix acting on $\mathcal H$; the dimension of the environment, $\dim \mathcal H_E$, is a parameter of the model. It was recognized in \cite{sommers2004statistical,nechita2007asymptotics} that the induced measure can be understood as a \emph{normalized Wishart measure}: 
$$\rho = \frac{Y}{\Tr Y},$$
where $Y = GG^*$, for $G : \mathcal H \to \mathcal H_E$ a random Gaussian matrix (an element of the \emph{Ginibre ensemble}).\\

We start by introducing the random matrix model for fermionic quantum states that we call the \emph{fermionic induced ensemble}. We then show the entanglement results for this model of random fermionic states.\\

In practice, we study an antisymmetric random Gaussian tensor $G_a\in \wedge^r \mathcal{H}\otimes \mathcal{H}_E$. We identify the space of antisymmetric tensors 
\begin{equation}\label{eq:identif_fermions}
    \wedge^r\mathcal{H}\cong \textrm{Ran}(P_{a,r})\subseteq \mathcal{H}^{\otimes r},
\end{equation}
where $P_{a,r}$ is defined in Section \ref{sec:first-def} to be the projector on the antisymmetric subspace of $\mathcal{H}^{\otimes r}$. We use this identification to construct the antisymmetric random Gaussian tensor from a complex random standard Gaussian tensor $G \in \mathcal{H}^{\otimes r}\otimes \mathcal{H}_E$, that is a tensor whose entries are standard normal random variables. We will be interested in the marginals obtained by tracing out the environmental degrees of freedom. We study the random matrix $A:= \Tr_E(G_a G_a^*)$. Note that, using the identification from \eqref{eq:identif_fermions}, we can also write
\begin{equation}\label{eq:fermionic-Wishart}
    A=P_{a,r}\Tr_E(GG^*)P_{a,r},
\end{equation}
where $G$ is the complex standard Gaussian tensor. In the language of quantum information theory, normalizing matrices $A$ yields random density matrices having an \emph{antisymmetry} constraint, reflecting their fermionic character on the $r$ copies of $\mathcal H$. 
\begin{definition}
 The \emph{fermionic induced ensemble} of parameters $(N,M,r)$ is the set of random bipartite density matrices of size $N^r$ defined by 
 \begin{equation}
     \rho^{(A)}=\frac{A}{\Tr A},
 \end{equation}
 where $A$ is an antisymmetrized Wishart random matrix from Eq.~\eqref{eq:fermionic-Wishart}. Matrices from the fermionic induced ensemble are supported on the antisymmetric subspace of $(\mathbb{C}^{N})^{\otimes r}$.
\end{definition}

Let us now comment on our choice of model. As already detailed in the introduction, our main motivation is that these types of states are more realistic than any (possibly non-symmetric) random states, since quite often in quantum experiments we consider states of electrons or say ${}^{87}\mathrm{Rb}$ (which are fermions), or states of photons or ${}^4\mathrm{He}$ atoms (which are bosons, whose ensemble we  introduce later in section \ref{sec:bosonic}). The corresponding degrees of freedom must have antisymmetric/symmetric wave functions, since they are fermions/bosons, while the environment has no reason to have any kind of symmetries. This explains why we impose the symmetry condition on the tensor product $\mathcal{H}^{\otimes r}$, while no symmetry is imposed on the tensor product with the environment system $\mathcal H_E$. \\

In the rest of this section, we set $r=2$. This is for technical reasons. If in the fermionic case, we could rather easily extend our results to general $r$, it does not appear feasible with our current techniques to study the bosonic case for general $r$. In order to stay coherent all along the paper we decided to present our results in the $r=2$ case for both fermionic and bosonic induced ensembles.\\

We set the dimension of the Hilbert spaces as follows
\begin{align*}
    \dim \mathcal{H}&=:N \\
    \dim \mathcal{H}_E&=:M.
\end{align*}
We shall consider an asymptotic regime where both $N,M \to \infty$ in such way that $M = cN^2$ for some fixed constant $c \in (0, \infty)$. This choice of scaling is the standard one needed for the convergence of the standard Wishart random matrices to the  Mar\u{c}enko-Pastur distribution:
\begin{equation}\label{eq:Marchenko-Pastur}
\mathrm{d}\mathrm{MP}_c=\max (1-c,0)\delta_0+\frac{\sqrt{(b-x)(x-a)}}{2\pi x} \; \mathbf{1}_{[a,b]}(x) \, \mathrm{d}x,
\end{equation}
with $a = (1-\sqrt c)^2$ and $b=(1+\sqrt c)^2$.\\

The limiting eigenvalue distribution of the properly normalized random matrix $A$ is the same as that of a non-symmetrized Wishart matrix. 

\begin{proposition}\label{prop:moments-A}
   The random matrix $A$ converges in moments, as $N \to \infty$, $M \sim cN^2$, to a  Mar\u{c}enko-Pastur distribution of parameter $2c$:
   $$\forall p \geq 1, \qquad \lim_{N \to \infty} \E \frac{1}{N_a[2]} \Tr \left[ \left(\frac{A}{N_a[2]}\right)^p \right] = \int x^p \mathrm{d}\mathrm{MP}_{2c}(x).$$
\end{proposition}
\begin{proof}
After restricting it to its support ($\Sym_2(\mathbb C^N)$), the random matrix $W$ is indeed a Wishart matrix of parameters $(N_a[2], M)$. The result follows from the standard theorem of Mar\u{c}enko and Pastur, see \cite{marcenko1967distribution}, \cite[Theorem 3.6]{bai2010spectral}, or \cite[Proposition 2.4]{dartois2020joint}. The parameter of the limiting distribution is given by 
$$\lim_{N \to \infty} \frac{M}{N_a[2]} = 2c.$$
\end{proof}

In order to understand the entanglement properties of the fermionic random states we just introduced, we proceed by first studying the average of their partial transposition. We show that $A^{\Gamma}$ has, on average, a negative eigenvalue, on a larger scale than the rest of the spectrum, for all $c>0$ (in the asymptotic regime where $M \sim cN^2$). Before we commence, recall that the \emph{maximally entangled state} is the unit norm vector
$$\ket \Omega = \frac{1}{\sqrt N} \sum_{i=1}^N \ket{ii} \in \mathbb C^N \otimes \mathbb C^N,$$
where $\{\ket i \}_{i \in [N]}$ is the canonical orthonormal basis of $\mathbb C^N$. We also write $\omega = \ketbra{\Omega}{\Omega}$ for the corresponding density matrix. 
\begin{proposition}\label{prop:fermionic-large-eigenvalue}
   The average of the partial transpose $A^{\Gamma}$ is given by
   \begin{equation}
       \E A^{\Gamma}=MP_{a,2}=\frac{M}{2}I_{N^2}-\frac{MN}{2} \omega.
   \end{equation}
   The eigenvalues of $\E A^{\Gamma}$ are as follows:
   \begin{equation}
       \mathrm{spec}(\E A^{\Gamma})=\begin{cases} -\frac{M(N-1)}{2} & \textrm{ with multiplicity } 1 \\
       \frac{M}{2} & \textrm{with multiplicity } N^2-1.
       \end{cases}
   \end{equation}
\end{proposition}
\begin{proof}
The proof is a direct computation. Using Wick-Isserlis theorem we compute $\E A=\frac{M}{2}(I_{N^2}-F)$, where $F$ is the flip (or swap) operator
$$F=\sum_{i,j=1}^N\ketbra{ij}{ji}.$$
Applying the partial transpose to the flip operator, we deduce that
$$F^{\Gamma}=\sum_{i,j=1}^N(\ketbra{ij}{ji})^{\Gamma}=\sum_{i,j=1}^N\ketbra{ii}{jj}=N\ketbra{\Omega}{\Omega} = N \omega.$$
Writing $P_{a,2} = (I-F)/2$ yields the formula for $\E A^\Gamma$.
The exact form of the spectrum follows from the fact that $\omega$ is a rank-one projection.
\end{proof}
We take away from the above proof that the maximally entangled vector $\lvert \Omega \rangle$ is an eigenvector of $\E A^{\Gamma}$ with a negative eigenvalue for all $N\ge2, c>0$, asymptotically of the form $-\frac{c}{2}N^3+O(N^2)$. The next theorem shows that, the same holds asymptotically for the random matrix $A^\Gamma$.

\begin{theorem}
The maximally entangled vector $\lvert \Omega \rangle$ is an approximate eigenvector of the random matrix $N^{-3}A^\Gamma$ with corresponding approximate eigenvalue $-c/2$: if  
\begin{equation}
    \delta_N:=\|N^{-3}A^\Gamma\lvert \Omega \rangle+c/2\lvert \Omega \rangle \|^2,
\end{equation}
then $\lim_{N\rightarrow\infty}\E\delta_N=0$. In particular, the random variable $\delta_N$ converges in probability towards zero as $N\rightarrow \infty$.
\end{theorem}
\begin{proof}
We start with a more explicit formula of $\delta_N$:
\begin{equation}
    \delta_N=N^{-6}\langle \Omega\rvert (A^\Gamma)^2\lvert \Omega\rangle+cN^{-3}\langle \Omega \rvert A^{\Gamma} \lvert \Omega \rangle+\frac{c^2}{4}.
\end{equation}
Computing the average, we obtain, thanks to Proposition \ref{prop:fermionic-large-eigenvalue},
\begin{equation}
    \E\delta_N=N^{-6}\E\langle \Omega\rvert (A^\Gamma)^2\lvert \Omega\rangle-\frac{c^2}{4}+o(1).
\end{equation}
We then use Wick-Isserlis theorem to compute $\E\langle \Omega\rvert (A^\Gamma)^2\lvert \Omega\rangle$, and we find an expression involving the partial trace $\Tr_1 P_{a,2}$,
\begin{align*}
    \E\langle \Omega\rvert (A^\Gamma)^2\lvert \Omega\rangle&= M^2\langle \Omega | (P_{a,2}^\Gamma)^2| \Omega \rangle +  \frac{M}{N} \Tr\left[ (\Tr_1 P_{a,2})^2 \right] \\
&=\frac{M^2}{N}\Tr\left[ (\Tr_1 P_{a,2})^2 \right]+\frac{M}{N} \Tr\left[ (\Tr_1 P_{a,2})^2 \right]\\
\end{align*}
We compute $\Tr_1(P_{a,2})=\frac12(N-1)I_N$ from which we deduce
\begin{align*}
    \E\langle \Omega\rvert (A^\Gamma)^2\lvert \Omega\rangle&=\left[ \frac{c^2}{4}N^6+\frac{c}{4}N^4\right]\left( 1-\frac2{N}+\frac1{N^2}\right)=\frac{c^2}{4}N^6+O(N^5),
\end{align*}
showing that $\E\delta_N\to 0$. We then use Markov's inequality to bound the probability that $\delta_N$ is positive. Indeed, for all $a>0$, we have
\begin{equation}
    \mathbb{P}(\delta_N\ge a)\le \frac1{a}\E\delta_N \to 0,
\end{equation}
concluding the proof.
\end{proof}
These two results are here sufficient to conclude that the random matrix $A^{\Gamma}$ is not positive semi-definite, as indeed it has a negative eigenvalue at scale $N^3$. As a consequence, the corresponding normalized random state $\rho^{(A)}$ is entangled. More precisely, we have the corollary
\begin{corollary}\label{cor:fermionic-threshold}
Consider a sequence $(\rho^{(A)}_N)_{N\ge2}$ of density matrices from the fermionic induced ensemble of parameters $(N,M,2)$, with $M=cN^2$. Then, for all $c>0$,
\begin{equation}
    \lim_{N\rightarrow \infty} \mathbb{P}[(\rho^{(A)}_N)^\Gamma\ge 0]=0,
\end{equation}
hence,
\begin{equation}
    \lim_{N\rightarrow \infty} \mathbb{P}[\rho^{(A)}_N \textrm{ is entangled}\,]=1.
\end{equation}
\end{corollary}
This is in contrast with the result of Aubrun \cite{aubrun2012partial} who considered random distinguishable states. In his framework, it was shown that the analog of the above statement is valid if and only if the environmental Hilbert space is large enough (\textit{i.e.} for $c>4$). What we show here is that for states of fermions, typical states are always entangled, this does not depend on the size of the environment as long as it is commensurable with the two particle Hilbert space $\wedge^2\mathcal{H}$ size.  In the next part we study the same problem for random  states of bosons. Bosonic states prove to be considerably more difficult to study and the rest of this paper is devoted to them.

\section{The bosonic induced ensemble}\label{sec:bosonic}

In this section, we shall consider a setting analog to the one of the previous section \ref{sec:fermionic}, only now the system Hilbert space has a \emph{symmetric tensor product structure}. In other words, we shall start from a random Gaussian tensor $G_s  \in \Sym_r(\mathcal{H})\otimes \mathcal{H}_E$. In practice, we shall use the identification 
\begin{equation}\label{eq:sym-subspace-range-Psym}
\Sym_r (\mathcal{H}) \cong \operatorname{Ran} (P_{s,r}) \subseteq \mathcal{H}^{\otimes r},   
\end{equation}
where we recall that $P_{s}$ is the orthogonal projection on the symmetric subspace of $\mathcal H \otimes \mathcal H$. The tensor factor $\mathcal{H}_E$ is again the environment Hilbert space. These random vectors $G_s$ can be seen as (un-normalized) bosonic states on two copies of $\mathcal H$ in contact with an environment $\mathcal H_E$. Similarly, we will be interested in the marginals obtained by tracing out the environmental degrees of freedom. We study the corresponding random matrix $W=\Tr_E(G_sG_s^*)$. Again, using the identification from \eqref{eq:sym-subspace-range-Psym}, we can also write
\begin{equation}\label{eq:def-W}
W=P_{s,r}\Tr_E(G G^*)P_{s,r},    
\end{equation}
where $G$ is a random (complex) Gaussian vector in the Hilbert space $\mathcal{H}^{\otimes r}\otimes \mathcal{H}_E$; see Figure \ref{fig:W} for a pictorial representation of the random matrix $W$ at $r=2$. Normalizing matrices $W$ yields random density matrices which have a \emph{symmetry} constraint. 

\begin{definition}\label{def:bosonic-induced}
The \emph{bosonic induced ensemble} of parameters $(N,M,r)$ is the set of random bipartite density matrices of size $N^r$ defined by
$$\rho = \frac{W}{\Tr W},$$
where $W$ is the symmetrized Wishart random matrix from \eqref{eq:def-W}. Matrices from the bosonic induced ensemble are supported on the symmetric subspace of $(\mathbb{C}^N)^{\otimes r}$.
\end{definition}

\begin{figure}
    \centering
    \includegraphics{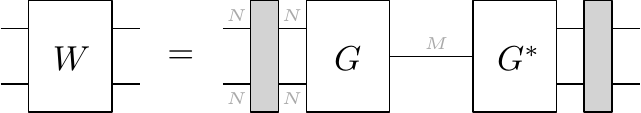}
    \caption{A graphical representation of the symmetrized Wishart matrix $W$ that we shall study in this work. The grey boxes surrounding the matrix $GG^*$ are symmetrizers.}
    \label{fig:W}
\end{figure}

Mimicking the construction of Section \ref{sec:fermionic}, we set the dimension of the Hilbert spaces as follows
\begin{align*}
    \dim \mathcal{H}&=:N \\
    \dim \mathcal{H}_E&:=M.
\end{align*}

Doing so, the limiting eigenvalue distribution of the properly normalized random matrix $W$ is the same as that of a non-symmetrized Wishart and also an anti-symmetrized matrix. 

\begin{proposition}\label{prop:moments-W}
   The random matrix $W$ converges in moments, as $N \to \infty$, $M \sim cN^2$, to a  Mar\u{c}enko-Pastur distribution of parameter $2c$:
   $$\forall p \geq 1, \qquad \lim_{N \to \infty} \E \frac{1}{N_s[2]} \Tr \left[ \left(\frac{W}{N_s[2]}\right)^p \right] = \int x^p \mathrm{d}\mathrm{MP}_{2c}(x).$$
\end{proposition}
\begin{proof}
After restricting it to its support ($\Sym_2(\mathbb C^N)$), the random matrix $W$ is indeed a Wishart matrix of parameters $(N_s[2], M)$. The result follows from the standard theorem of Mar\u{c}enko and Pastur, see \cite{marcenko1967distribution}, \cite[Theorem 3.6]{bai2010spectral}, or \cite[Proposition 2.4]{dartois2020joint}. The parameter of the limiting distribution is given by 
$$\lim_{N \to \infty} \frac{M}{N_s[2]} = 2c.$$
\end{proof}

\begin{remark}
   The exact moments of order $p$ of the random matrix $W$ can be computed in different manners, either as a sum over permutations, or as a sum over bipartite combinatorial maps 
   with one black vertex and $p$ edges $\mathcal{M}=(\alpha, \gamma)$ where $\alpha$ is a permutation of $S_p$ and $\gamma = (1 \, 2 \, 3\, \cdots \, p)$ is the full cycle:
   \begin{align*}
       \E  \Tr \left[ \left(\frac{W}{N_s[2]}\right)^p \right] &= \sum_{\alpha \in S_p} N_s[2]^{\# (\gamma^{-1}\alpha)} M^{\#\alpha} \\
       &= \sum_{\mathcal{M}}N_s[2]^{F(\mathcal{M})}M^{V_{\circ}(M)},
   \end{align*}
   where $F(\mathcal{M})$ is the number of faces of $\mathcal{M}$ and $V_{\circ}(\mathcal{M})$ its number of white vertices. We refer the reader to \cite[Section 2]{dartois2020joint} for the dictionary between these two approaches. 
\end{remark}

\begin{remark}
   The asymptotic moments from Proposition \ref{prop:moments-W} correspond to sums over \emph{geodesic permutations} $\alpha$, or, equivalently, over \emph{planar bipartite maps, or dessins d'enfants,} $\mathcal M$ with one black vertex. The limiting moments are values of the {Narayana polynomial}:
   $$\mathrm{Nar}(p, 2c) = \int x^p \mathrm{d}\mathrm{MP}_{2c}(x) = \sum_{k=0}^{p-1}\frac{1}{(k+1)(2c)^k} \binom{p-1}{k}\binom{p}{k}.$$
\end{remark}

\section{The partial transposition --- statement of main results}\label{sec:main-result}

In this section we state the main result of this work, the characterization of the limiting spectrum of the \emph{partial transposition} of the random matrix $W$, in the large $N$ limit. Recall that $W$ is a $N^2 \times N^2$ random matrix, supported on the symmetric subspace $\Sym_2(\mathbb C^N) \subseteq \mathbb C^N \otimes \mathbb C^N$. Its partial transpose is defined as the action of the transposition operator on the second tensor factor: 
\begin{equation}\label{eq:def-W-Gamma}
    W^\Gamma = [\mathrm{id} \otimes \mathrm{transp}](W).
\end{equation}

The next theorem is the main result of our paper, providing a description of the spectrum of $W^\Gamma$ in the large $N$ limit. Interestingly, $W^\Gamma$ has a large eigenvalue of order $N^3$, and $N^2-1$ eigenvalues of order $N^2$.

\begin{theorem}\label{thm:main}
Let $\lambda_1\ge \lambda_2\ge \ldots \lambda_{N^2}$ be the eigenvalues of the random matrix $W^\Gamma$ from Eq.~\eqref{eq:def-W-Gamma}. Then, in the asymptotic regime $N \to \infty$ and $M \sim cN^2$,
\begin{itemize}
    \item $N^{-3}\lambda_1 \rightarrow \frac{c}{2}$ in probability
    \item The empirical distribution $\sum_{i=2}^{N^2}\delta_{N^{-2}\lambda_i}$ converges in moments to a semi-circular distribution of mean $c/2$ and variance $c/4$. 
\end{itemize}
\end{theorem}

We recall here that the \emph{semicircular distribution} with average $m$ and variance $\sigma^2$ is given by $$\mathrm{d} \mathrm{SC}_{m,\sigma} = \frac{1}{2 \pi \sigma} \sqrt{4\sigma^2 - (x-m)^2} \mathbf{1}_{[-2\sigma+m, 2\sigma+m]}(x) \, \mathrm{d}x.$$

Let us record here an important corollary for the theory of quantum information. The \emph{partial transposition criterion} \cite{peres1996separability,horodecki1996separability} states that a \emph{separable} (i.e.~non-entangled) density matrix $\rho$ has a positive semidefinite partial transposition: $\rho^\Gamma \geq 0$. This fact is mostly used in the reverse direction, as an entanglement criterion: given a density matrix $\sigma$ having a partial transposition with at least one negative eigenvalue ($\sigma^\Gamma \ngeq 0$) is entangled. We have thus the following result, regarding the bosonic induced ensemble defined in Definition \ref{def:bosonic-induced}. 

\begin{corollary}\label{cor:bosonic-threshold}
Consider a sequence $(\rho_N)_{N \geq 1}$ of density matrices from the bosonic induced ensemble of parameters $(N,M,2)$, with $M \sim cN^2$ as $N \to \infty$. If $c < 4$, then
$$\lim_{N \to \infty} \P[\rho_N^\Gamma \geq 0] = 0.$$
In particular, 
$$\lim_{N \to \infty} \P[\rho_N \text{ is entangled }] = 1.$$
\end{corollary}

\begin{remark}
   Note that the other regime, where $c  > 4$, is much more involved, since the absence of negative support of the limiting measure of $\rho_N^\Gamma$ does not exclude the absence of negative outliers. In the standard (non-symmetrized) Wishart case, Aubrun \cite{aubrun2012partial} computed large moments of the partial transposition in order to achieve this conclusion. We leave these considerations open. Indeed, in our case, considerations on large moments would not allow us to conclude. This is due to our use of the Cauchy interlacing theorem. The study of large moments of the compressed matrix $Q$, introduced later, does not allow us to retain enough information on the behaviour of outliers of $\rho_N^\Gamma$; at the same time, in this work, we cannot study them directly because of the large eigenvalue of $W^\Gamma$, as will be seen later.
\end{remark}

\bigskip

The proof of the main theorem is rather involved, and we shall proceed in several steps in the following sections, culminating with the final proof given in Section \ref{sec:proof-main-result}. We start, in the spirit of Section \ref{sec:fermionic}, with an analysis of the average state and of the outlier eigenvalue. However, in the bosonic case, this will not be sufficient to make a statement about entanglement. Hence the need of a more detailed study of the spectrum, provided by our main result, Theorem \ref{thm:main}. 

Following the approach of Section \ref{sec:fermionic},  we compute the average of the partially transposed matrix $W^\Gamma$, and show that already on average, there is an outlier, \emph{positive}, eigenvalue on a scale larger than the rest of the spectrum.
Recall that the \emph{maximally entangled state} is the unit norm vector
$$\ket \Omega = \frac{1}{\sqrt N} \sum_{i=1}^N \ket{ii} \in \mathbb C^N \otimes \mathbb C^N,$$
where $\{\ket i \}_{i \in [N]}$ is the canonical orthonormal basis of $\mathbb C^N$.

\begin{proposition}\label{prop:eigenval-eigenvec-average}
The average of the partial transpose $W^{\Gamma}$ is given by
$$\E W^\Gamma = M P_s^\Gamma = \frac{M}{2} I_{N^2} + \frac{MN}{2} \ketbra{\Omega}{\Omega}.$$
The eigenvalues of $\E W^\Gamma$ are as follows
$$\operatorname{spec}(\E W^\Gamma) = \begin{cases}
\frac{M(N+1)}{2} \qquad &\text{ with multiplicity $1$}\\
\frac{M}{2} \qquad &\text{ with multiplicity $N^2-1$.}
\end{cases}
$$
\end{proposition}
\begin{proof}
The proof is a direct computation. Using Wick-Isserlis theorem we compute $\E W=\frac{M}{2} (I_{N^2}+F)$. Applying the partial transposition on the flip operator
$$F=\sum_{i,j=1}^N\ketbra{ij}{ji},$$
we deduce that
$$F^{\Gamma}=\sum_{i,j=1}^N(\ketbra{ij}{ji})^{\Gamma}=\sum_{i,j=1}^N\ketbra{ii}{jj}=N\ketbra{\Omega}{\Omega}.$$
Writing $P_s = (I+F)/2$ yields the formula for $\E W^\Gamma$.
The spectrum follows from the fact that $\ketbra{\Omega}{\Omega}$ is a rank-one projection.
\end{proof}

We note that $\lvert \Omega \rangle$ is an eigenvector of $P_s^{\Gamma}$ with eigenvalue $(N+1)/2$. Hence, the average $\E W^\Gamma$ has an eigenvalue $M(N+1)/2\sim cN^3/2$ with eigenvector $\ket \Omega$. We show next that the same holds, asymptotically, for $W^\Gamma$.
\begin{theorem}\label{thm:eigenvec-convergence}
The maximally entangled vector $\ket \Omega$ is an approximate eigenvector of the random matrix $N^{-3}W^\Gamma$ with corresponding approximate eigenvalue $c/2$: if 
$$\epsilon_N:=\|N^{-3} W^\Gamma \ket \Omega - c/2 \ket \Omega\|^2,$$
then $\lim_{N \to \infty} \E \epsilon_N = 0$. In particular, the random variable $\epsilon_N$ converges in probability towards zero as $N\to\infty$.
\end{theorem}
\begin{proof}
Manipulating the expression of $\epsilon_N$ we have 
$$\epsilon_N=N^{-6}\langle \Omega | (W^\Gamma)^2 | \Omega \rangle  - cN^{-3} \langle \Omega | W^\Gamma | \Omega \rangle+ c^2/4.$$
Computing the average and using Proposition \ref{prop:eigenval-eigenvec-average} we have 
$$\E \epsilon_N=N^{-6} \E\langle \Omega | (W^{\Gamma})^2| \Omega \rangle - c^2/4 + o(1).$$

We compute $\E\langle \Omega | (W^{\Gamma})^2| \Omega \rangle$ using the Wick-Isserlis theorem and find, using Lemma \ref{lem:partial-trace-P-sym},
\begin{align*}
\E\langle \Omega | (W^{\Gamma})^2| \Omega \rangle &= M^2\langle \Omega | (P_s^\Gamma)^2| \Omega \rangle +  \frac{M}{N} \Tr\left[ (\Tr_1 P_s)^2 \right] \\
&=\frac{M^2}{N}\Tr\left[ (\Tr_1 P_s)^2 \right]+\frac{M}{N} \Tr\left[ (\Tr_1 P_s)^2 \right]\\ &=\left[N^6\frac{c^2}{4}+N^4\frac{c}{4}\right]\left(1+\frac2N+\frac1{N^2}\right)= \frac{c^2}{4}N^6+O(N^5),
\end{align*}
establishing the first claim. 
The convergence in probability follows by Markov's inequality.
\end{proof}

\bigskip

Let us finish this section by addressing another important entanglement criterion: \emph{realignment} \cite{chen2002matrix, rudolph2003cross}. The realignment criterion states that any separable bipartite density matrix $\rho$ satisfies the following inequality: 
$$\|\rho^R\|_1 \leq 1,$$
where $\|\cdot \|_1$ denotes the Schatten 1-norm or the nuclear norm (i.e.~the sum of the singular values of a matrix), while $\rho^R$ denotes the realignment (or the reshuffling, see \cite[Chapter 10.2]{bengtsson2006geometry}) of a matrix $\rho$, defined algebraically by
$$\langle ij | \rho^R | kl \rangle = \langle ik | \rho | jl \rangle,$$
or graphically in Figure \ref{fig:realignment}. 

\begin{figure}
    \centering
    \includegraphics{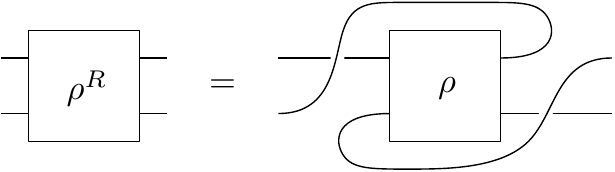}
    \caption{The realignment of a bipartite matrix $\rho$.}
    \label{fig:realignment}
\end{figure}

The key observation here is that the realignment and the partial transposition of a matrix are related as follows: 
$$\rho^R = \left(\rho^\Gamma\right)\cdot F \quad \implies \quad \|\rho^R\|_1 = \|\rho^\Gamma\|_1.$$

In the case of a (non-normalized) random matrix $W$ from the induced bosonic ensemble, entanglement follows from the inequality $\|W^\Gamma\|_1 > \Tr W$. Note that, asymptotically as $N \to \infty$ and $M \sim cN^2$, the trace behaves as $\Tr W \sim (c/2) N^4$, while Theorem \ref{thm:main} states that
$$\|W^\Gamma\|_1 \sim (c/2) N^3 + N^4 \int |x| \, \mathrm{d} \mathrm{SC}_{c/2, \sqrt c/2}(x).$$
We note that the large outlier eigenvalue does not contribute asymptotically to the value of the Schatten 1-norm. Hence, in order for the realignment criterion to detect entanglement in the random bosonic matrix $W$, the following inequality needs to hold: 
$$\int |x| \, \mathrm{d} \mathrm{SC}_{c/2, \sqrt c/2}(x) > \frac c 2 = \int x \, \mathrm{d} \mathrm{SC}_{c/2, \sqrt c/2}(x).$$
But the above happens precisely when the probability measure $\mathrm{SC}_{c/2, \sqrt c/2}$ has negative support, which is also the condition for the partial transpose criterion to detect the entanglement of $W$. Hence, we conclude that, in the case of the induced bosonic ensemble, the partial transpose and the realignment criteria have \emph{identical thresholds} ($c_0 = 4$). This situation is to be contrasted with the case of the usual induced ensemble, where the threshold for the partial transpose criterion ($c^\Gamma_0 = 4$ \cite{aubrun2012partial}) is strictly larger than the threshold for the realignment criterion ($c_0^R = (8/3\pi)^2 \approx 0.72$ \cite{aubrun2012realigning}), proving that, in some sense, the former criterion is stronger than the latter. 

\section{The circuit counting graph polynomial}\label{sec:graph-poly}

In this subsection we discuss some basic facts about the \emph{circuit counting polynomial} of directed graphs. Of special importance to us is Theorem \ref{thm:graph-poly-trace-Psym}, relating the generalized bipartite trace of the symmetric projector $P_s = P_{s, 2}$ to the circuit counting polynomial $J$ of an associated graph.

We start by recalling the definition of a digraph. 
\begin{definition}
 A digraph is a pair $G=(V,A)$ where $V$ is a set of vertices and $A$ is a multi-set of ordered pairs of (possibly non distinct) elements of $V$. The fact that $A$ is a multi-set allows for multiple edges.
\end{definition}
In this paper, we will denote when needed the elements of $A$ as $\{i\rightarrow j\}$ for $i,j \in V$ to better picture the ordering. Possibly non-distinct elements in the pairs means the digraphs can have loops. \\
Note that incoming edges of a vertex $i$ are the edges of the form $\{a\rightarrow i\}$, while outgoing edges are the edges of the form $\{i\rightarrow a\}$. A loop edge is both incoming and outgoing for the vertex it is adjacent to. A digraph is said to be $k$-in/$k$-out if each vertex has $k$ incoming edges and $k$ outgoing edges.\\

We introduce the circuit counting polynomial $J$ of a digraph $G$ as
\begin{equation}
    J(G;x)=\sum_{k\ge 1}j_kx^k,
\end{equation}
with $j_k=|\{ \textrm{covers of } G \textrm{ in } k \textrm{ cycles}\}|$. Note that the cycles must follow the edge orientations.
The circuit counting polynomial possesses several properties useful to us. As we will be interested only in the case of digraphs that are $2$-in/$2$-out, we state these properties only in the case of $2$-in/$2$-out digraphs. The circuit counting polynomial has the following trivial multiplicativity property. 

\begin{lemma}\label{lem:mult-J}
Let $G$ be a 2-in/2-out digraph having connected components $G_1, G_2, \ldots, G_k$ (which are connected 2-in/2-out digraphs). Then
$$J(G; x) = \prod_{i=1}^k J(G_i; x).$$
\end{lemma}

Next, we consider skein relations. The reason is that a cycle cover of $G$ corresponds to a choice of state for each vertex of $G$ (as defined in \cite[Definition 2.1]{ellis2004identities} and below; the case we are interested in is the Eulerian case). These skein relations read graphically
\begin{equation}\label{eq:skein-relation-J}
    J\left(\raisebox{-16mm}{\includegraphics[scale=0.4]{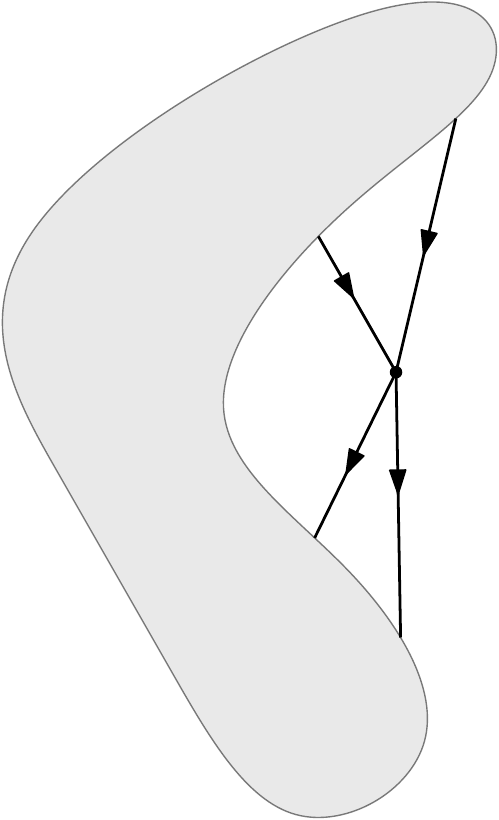}};x\right) \ = \ J\left(\raisebox{-16mm}{\includegraphics[scale=0.4]{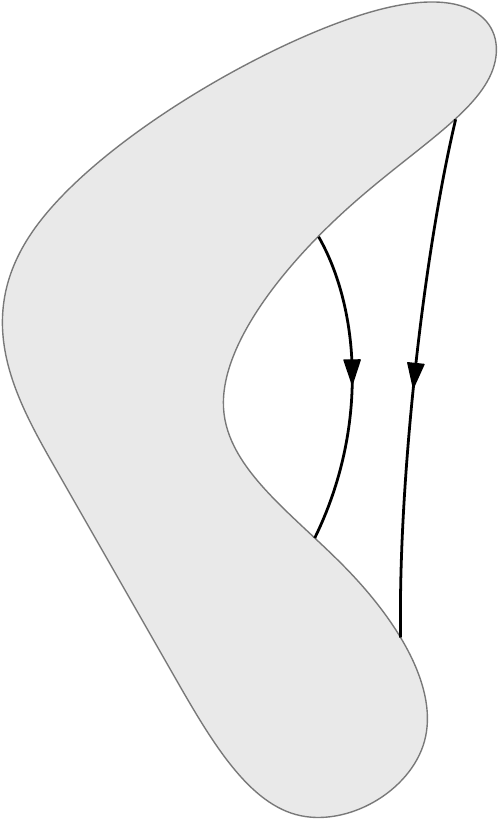}};x\right) \ + \ J\left( \raisebox{-16mm}{\includegraphics[scale=0.4]{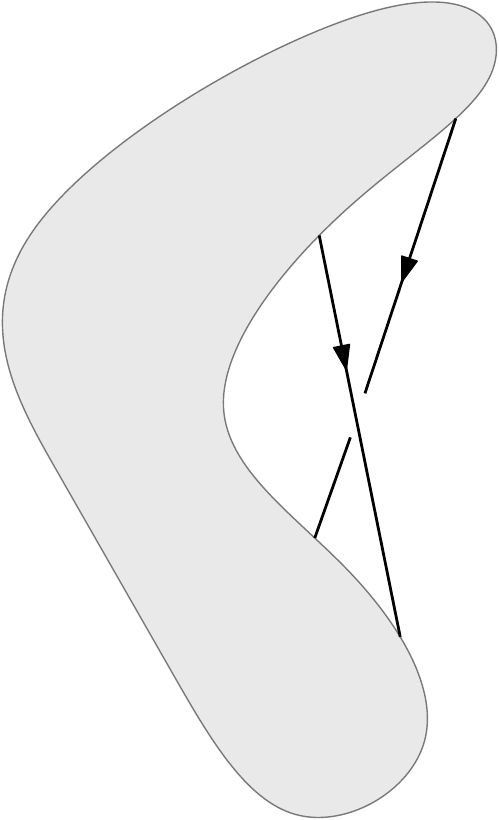}};x\right).
\end{equation}
Additionally, when evaluated on connected graphs, it is of maximal degree for graphs that have only cut vertices, where we recall that a cut vertex is a vertex whose removal disconnects the graph. This fact is a simple extension of \cite{las1983polynome}, that we show here. We give a proof below, using the graphical skein relations.

\begin{theorem}\label{thm:cut-vertex-relation}
Let $G$ be a $2$-in/$2$-out connected digraph and $v$ be a cut vertex. Denote by $G\backslash v$ the \emph{connected} digraph obtained from $G$ by erasing the vertex $v$ and reconnecting the edges in an orientation preserving way. Then 
\begin{equation}\label{eq:cut-vertex-relation}
    J(G;x)=(x+1)J(G\backslash v;x).
\end{equation}
\end{theorem}
\begin{proof}
We have using the skein relations described above at the cut vertex $v$
\begin{equation}\label{eq:move-J}
    J\left(\raisebox{-10mm}{\includegraphics[scale=0.4]{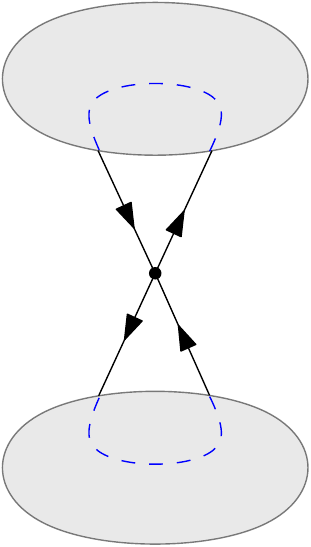}};x\right)\ = \ J\left(\raisebox{-10mm}{\includegraphics[scale=0.4]{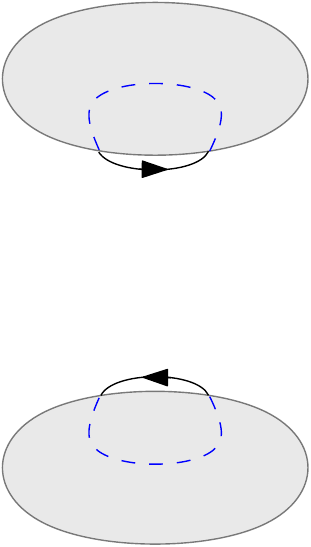}};x\right) \ + \ J\left(\raisebox{-10mm}{\includegraphics[scale=0.4]{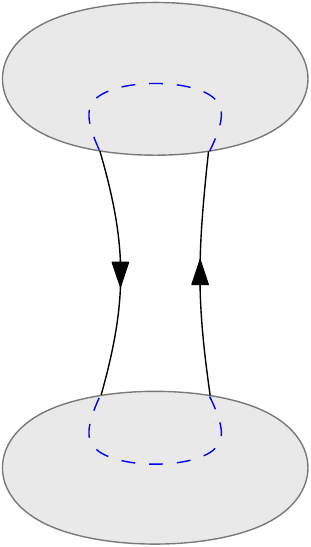}};x\right)
\end{equation}
where the gray areas represent the rest of the graph. Since $v$ is a cut vertex these gray areas are connected only through $v$. The blue dashed lines inside the gray parts show how a cycle that would pass through the edges entering and exiting the gray parts would behave. Notice in particular that due to the $2$-in/$2$-out constraint, a cycle going through one edge entering a gray part has to go out of this gray part through the edge exiting it. Now notice that one of the choice of state for $v$ (the leftmost term of the right hand side) leads to a disconnected graph that has one more cycle than the other choice of state that leads to a connected graph (rightmost term of the right hand side). Thus we have, graphically,
\begin{equation}
   J\left(\raisebox{-10mm}{\includegraphics[scale=0.4]{J-cut-vertex-move1.pdf}};x\right) = x J\left(\raisebox{-10mm}{\includegraphics[scale=0.4]{J-cut-vertex-move2.pdf}};x\right).
\end{equation}
Therefore,
\begin{equation}
    J\left(\raisebox{-10mm}{\includegraphics[scale=0.4]{J-cut-vertex.pdf}};x\right)\ = \ (x+1)J\left(\raisebox{-10mm}{\includegraphics[scale=0.4]{J-cut-vertex-move2.pdf}};x\right).
\end{equation}
\end{proof}
Let $G$ be  a $2$-in/$2$-out digraph with $p$ vertices such that if every vertex is a cut vertex. Then repeated use of equation \eqref{eq:cut-vertex-relation} leads to  
\begin{equation}\label{eq:only-cut-J}
    J(G;x)=(x+1)^pJ(L;x)=x(x+1)^p,
\end{equation}
where $L$ is the trivial digraph made of an edge looping on itself with no vertices. We later show that digraphs with only cut vertices are the ones maximizing the degree of the circuit counting polynomial.

In the case of $2$-in/$2$-digraphs, the circuit counting polynomial is related to the \emph{interlace polynomial} \cite{arratia2004interlace}, defined here according to \cite{brijder2011nullity}.  For a (non oriented) graph $H=(V,E)$, the interlace polynomial $q(H;x)$ is
\begin{equation}
    q(H;x)=\sum_{\substack{S\subseteq V}} (x-1)^{\textrm{dim }\textrm{ker } A_{H[S]}},
\end{equation}
where $H[S]$ is the sub-graph induced by the subset $S$ of vertices of $H$ and $A_{H[S]}$ denotes the adjacency matrix of $H[S]$. Note that $\textrm{dim }\textrm{ker } A_{H[S]}$ is often called the nullity of $H[S]$. Note also that the degree of the interlace polynomial of a graph is bounded by its number of vertices:
\begin{equation}\label{eq:UB-deg-q}
    \deg q(H; x) \leq |H|.
\end{equation}

We use repeatedly the following property. For connected 2-in/2-out digraphs, the circuit polynomial is related to the interlace polynomial $q(x)$ of a corresponding interlace graph via the following theorem.

\begin{theorem}[{\cite[Theorem 24]{arratia2004interlace}}]\label{thm:Arratia}
    Consider $G$ an Eulerian $2$-in/$2$-out digraph, $C$ any Euler circuit of $G$ and $H=H(C)$ its interlace graph. Then 
    \begin{equation}
        J(G;x)=x q(H;x+1)=x m(G;x+1)
    \end{equation}
    with $q$ the interlace polynomial, and $m$ the Martin polynomial.
\end{theorem}
Let us recall a bit of terminology. An Euler circuit $C$ of a connected digraph $G$ is a circuit passing through every edge of $G$ exactly once. For a $2$-in/$2$-out digraph $G$ the interlace graph $H(C)$ of one of its Euler circuits $C$ is the (non oriented!) graph defined as having the same vertex set than $G$ and having an edge between vertices $a$ and $b$ if those are interlaced in $C$. That is, they appear in the following way $C = \ldots a \ldots b \ldots a \ldots b \ldots$, see \cite{arratia2004interlace} for these definitions. It is easy to deduce that connected $2$-in/$2$-out digraphs $G$ having an Euler circuit with no interlace pairs have an interlace graph $H(C)$ maximizing the degree of $q(H(C);x)$ since $H(C)$ is empty (thus $\textrm{dim }\textrm{ker } A_{H(C)[S]}=|S|$ for all $S\subseteq V$).\\

In \cite{arratia2004interlace} the following result is established.
 \begin{lemma}
 If a $2$-in/$2$-out digraph $G$ has an Euler circuit $C$ with no interlace pairs, then $G$ has only one Euler circuit. 
 \end{lemma}
 A simple consequence is that $2$-in/$2$-out digraphs $G$ maximizing the degree of the interlace polynomial of their interlace graphs have a unique Euler circuit, see also Lemma \ref{lem:max-J-only-cut}. We can now give a very useful bound for the degree of the circuit counting polynomial. We denote by $K(G)$ the number of connected components of the graph $G$; here, we do not care about the orientation of the edges in $G$ when defining connected components.

\begin{lemma}\label{lem:UB-deg-J}
Given a 2-in/2-out digraph $G$ on $n$ vertices, we have 
$$\deg J(G;x) \leq n + K(G).$$
\end{lemma}
\begin{proof}
We shall use the relation between the circuit counting polynomial and the interlace polynomial from Theorem \ref{thm:Arratia} on the connected components $G_1, G_2, \ldots, G_{K(G)}$ of $G$. The degree of the interlace polynomial of a connected $2$-in/$2$-out digraph with $p$ vertices is bounded by $p$ (see Eq.~\eqref{eq:UB-deg-q}), hence $\deg J(G_i) \leq n_i+1$, where $n_i$ is the number of vertices of $G_i$. Using the multiplicativity property of $J$ from Lemma \ref{lem:mult-J}, we have
 $$\deg J(G; x) = \sum_{i=1}^{K(G)} \deg J(G_i; x) \leq \sum_{i=1}^{K(G)} n_i + 1 = n + K(G).$$
\end{proof}
We then use the above bound to show the next lemma.
\begin{lemma}\label{lem:max-J-only-cut}
   Among all connected digraphs with a given number of vertices, the ones maximizing the degree of the circuit-counting polynomial $J$ are precisely the ones having only cut vertices. 
\end{lemma}
\begin{proof}
Assume that a connected digraph $G$ has only $k<p$ out of $p$ vertices which are cut vertices. Then using skein relations \eqref{eq:skein-relation-J}, we can pick a vertex that is not a cut vertex of $G$ and reduce it. We obtain two connected $2$-in/$2$-out digraphs with $p-1$ vertices $G'$ and $G''$. Indeed, they are connected because the vertex we chose to reduce is not a cut-vertex, hence none of the moves disconnect the digraph. Then using the bound of Lemma \ref{lem:UB-deg-J} above we have that $\textrm{deg} \ J(G'), \textrm{deg} \ J (G'')\le p$. We also know from equation \eqref{eq:only-cut-J} that connected $2$-in/$2$-out digraphs with $p$ vertices have circuit counting polynomial of degree $p+1$, proving the claim.
\end{proof}

We introduce a practical notation for the rest of this paper. We denote $G_{\sigma_1,\sigma_2}$, for $\sigma_1, \sigma_2\in S_p$, the $2$-in/$2$-out digraph on $p$ vertices whose (oriented) edge (multi-)set is 
\begin{equation}\label{eq:def-G-a-b}
    E=\{i\rightarrow \sigma_1(i) \, : \,  i\in [p]\} \sqcup \{ i\rightarrow\sigma_2(i) \, : \, i \in [p]\}.
\end{equation}
Note in particular that there is an edge of the form ${i\rightarrow k}$ if either the matrix of the permutation $\sigma_1$ has element $(k,i)$ equal to $1$ or the matrix of the permutation $\sigma_2$ has element $(k,i)$ equal to $1$. If there are two edges of the form ${i\rightarrow k}$ then both matrices of the permutations $\sigma_1, \sigma_2$ have elements $(k,i)$ equal to $1$. Finally, there is no edge from $i$ to $k$ if and only if both $\sigma_1, \sigma_2$ have elements $(k,i)$ equal to $0$. These facts together with a moment of reflection reveal that the adjacency matrix $A_{G_{\sigma_1,\sigma_2}}$ of the $2$-in/$2$-out digraph is the sum of the two matrices of the permutations $\sigma_1, \sigma_2$
\begin{equation}
    A_{G_{\sigma_1,\sigma_2}}= \sigma_1+\sigma_2,
\end{equation}
where we used the same notation for the permutations $\sigma_{1,2}$ and their matrix representations. In general, we also have the reciprocal, that holds in a more general setting,
\begin{proposition}
   Given a $k$-in/$k$-out digraph $G$ with $p$ vertices, there exists $k$ permutations $\sigma_1,\sigma_2,\ldots, \sigma_k\in S_p$ such that 
   \begin{equation}
       A_{G}=\sum_{i=1}^{k}\sigma_i
   \end{equation}
   where $A_G$ is the adjacency matrix of the digraph $G$.
\end{proposition}
\begin{proof}
First notice that due to the regularity of the digraph $G$, $A_G$ is, up-to-normalization, a bistochatic matrix. Indeed, the $k$-in/$k$-out property implies that every row and every column of $A_G$ must sum to $k$. From this remark, the proposition is a trivial consequence of the Birkhoff-von Neumann algorithm applied to $A_G$.
\end{proof}
This proposition implies that by choosing to index $2$-in/$2$-out digraphs with two permutations we did not restrict the set of digraphs we are looking at. However note that the choice of permutations, given a digraph, is not unique and there may be different collections of permutations representing the same digraph. It is possible to describe classes of permutations that lead to equivalent digraphs, but this will be of no use to us. \\
\begin{figure}
    \centering
    \includegraphics[scale=0.8]{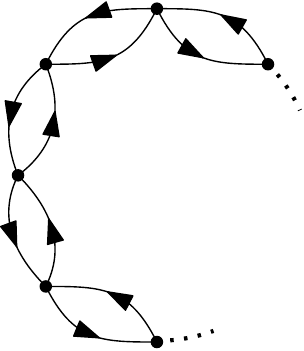}
    \caption{Structure of a digraph $G_{\gamma,\gamma^{-1}}$.}
    \label{fig:digraph-gamma-gamma-inv}
\end{figure}
Assume $\gamma=(1,2,3,\ldots, n)\in S_n$ is the full cycle permutation, then
\begin{proposition}\label{prop:J-alpha=id}
   Consider the $2$-in/$2$-out digraph $G_{\gamma, \gamma^{-1}}$, where $\gamma\in S_n$ is the full cycle. The number of covers of $G_{\gamma, \gamma^{-1}}$ having $k$ cycles is given by
   \begin{equation}
    j_k=\binom{n}{k}+\delta_{2,k}.
   \end{equation}
   In particular, 
   \begin{equation}
       \mathrm{deg }\, J(G_{\gamma, \gamma^{-1}})=n.
   \end{equation}
\end{proposition}
\begin{proof}
The proof follows from the fact that
$G_{\gamma,\gamma^{-1}}$ is of the form depicted in Fig \ref{fig:digraph-gamma-gamma-inv}. Then it is sufficient to notice that at each vertex of $G_{\gamma,\gamma^{-1}}$ we have the choice to either come back to where we came from (vertex state $1$) or to keep going in the same direction (vertex state $2$). The case for which we decide to keep going at all vertices (\textit{i.e.} we pick state $2$ at all vertices) leads to two cycles (responsible for the $\delta_{2,k}$ in the above formula). While if we decide to have one vertex with state $1$ and the others with state $2$ then we get one cycle. From that, a moment of reflection reveals that each time we assign an additional vertex the state $1$ we create an additional cycle. Thus there are $\binom{n}{k}$ choice of states leading to a cycle cover with $k$ cycles of $G_{\gamma,\gamma^{-1}}$. \\

It follows that the degree of  $J(G_{\gamma, \gamma^{-1}})$ is $n$. This last result can also easily be obtained by constructing an interlace graph $H_{\gamma,\gamma^{-1}}$ from an Eulerian cycle $C_{\gamma,\gamma^{-1}}$ of $G_{\gamma, \gamma^{-1}}$.
It is easy to exhibit such an Eulerian cycle: $C_{\gamma,\gamma^{-1}}=1,2,3,\ldots,n,1,n,\ldots,3,2$. A pair $(1,q)$ is clearly interlaced for any $q=2,\ldots, n$ while any other pair $(a,b)$ for $a \textrm{ and }b\neq 1$ is not interlaced. Thus $H_{\gamma,\gamma^{-1}}$ is the star graph with the vertex $1$ at its center and thus $\underset{{X\textrm{ minor of } H}}{\textrm{max }}\textrm{dim }\textrm{ker } X=n$ the maximizing minor being the graph obtained from $H$ by removing vertex $1$.
\end{proof}

Finally, we now relate the circuit counting graph polynomial $J$ to a generalized trace of the projector on the symmetric subspace $P_s$ (see also Proposition \ref{prop:expanded-Wick-thm} for a more general result obtained using the formalism of tensor networks). 

\begin{theorem}\label{thm:graph-poly-trace-Psym}
For all permutations $\alpha, \beta \in S_p$, we have
$$\Tr_{\alpha,\beta}(P_{s}, \ldots, P_{s}) = 2^{-p} J(G_{\alpha,\beta},N).$$
\end{theorem}
\begin{proof}
We start from $P_s = (I + F)/2$ and expand the left hand side of the equation in the statement as a sum over all possible choices of terms: 
$$\Tr_{\alpha,\beta}(P_{s}, \ldots, P_{s}) = 2^{-p} \sum_{f : [p] \to \{I,F\}} \Tr_{\alpha,\beta}\left(\bigotimes_{i=1}^p f(i) \right).$$

Above, the generalized trace of the right hand side, in diagrammatic notation, is a collection of loops evaluating to $N$ raised to the power of the number of loops. Clearly, each such loop corresponds to a circuit in the digraph $G$ given by the permutations $\alpha$ and $\beta$, once the moves $f(i)$ have been applied at each vertex $i \in [p]$, see \eqref{eq:move-J}.
\end{proof}

\section{Moments of \texorpdfstring{$W^\Gamma$}{W Gamma}}\label{sec:moments}

In this section, we shall express the moments of the random matrix $W^\Gamma$ (see Figure \ref{fig:W-Gamma}) as a sum of circuit counting polynomials, allowing us to determine their asymptotic behaviour. On the way, we shall establish several useful properties of the aforementioned polynomials, which shall also be useful in the later sections. The main result, Theorem \ref{thm:moments-WGamma-3-asympt} showcases the power of the relation between the generalized bipartite traces of $P_s$ and the circuit counting polynomial from Theorem \ref{thm:graph-poly-trace-Psym}. 

\begin{figure}
    \centering
    \includegraphics{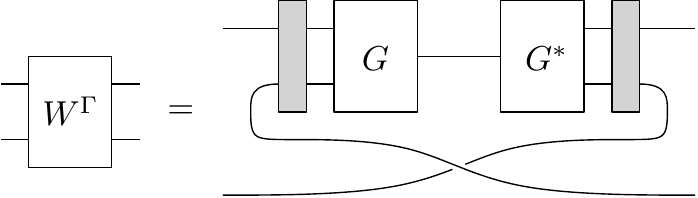}
    \caption{The partial transposition $W^\Gamma$ of a symmetrized Wishart matrix.}
    \label{fig:W-Gamma}
\end{figure}
We first establish the exact form of the moments of $W^\Gamma$ in terms of $J$ polynomials.

\begin{proposition}\label{prop:moments-W-Gamma}
   The moments of the random matrix $W^\Gamma \in \mathcal M_{N^2}(\mathbb C)$ from Eq.~\eqref{eq:def-W-Gamma} are given by:
\begin{equation}\label{eq:moments-W-Gamma}
\forall p \geq 1, \qquad     \E\Tr\left((W^{\Gamma})^p\right) =  2^{-p} \sum_{\alpha\in S_p}M^{\#\alpha}J(G_{\gamma\alpha, \gamma^{-1}\alpha};N).
\end{equation}
\end{proposition}
\begin{proof}
The proof is a rather simple application of the Wick formula, in its tensor network incarnation \cite{collins2011gaussianization}. We need to evaluate the expectation of the trace of the product of $p$ copies of the matrix from Figure \ref{fig:W-Gamma}. By the Wick theorem, the result is a sum over permutations $\alpha \in S_p$ of diagrams obtained by connecting the $i$-th $G$ box to the $\alpha(i)$-th $G^*$ box

Consider now the 2-in/2-out digraph $G_{\gamma\alpha, \gamma^{-1}\alpha}$ having edges (see \eqref{eq:def-G-a-b} for the general case)
\begin{equation}\label{eq:def-G-ga-gma}
    \{ i \mapsto \gamma(\alpha(i)) \}_{i \in p} \sqcup \{ i \mapsto \gamma^{-1}(\alpha(i)) \}_{i \in p}.
\end{equation}
The formula in the statement follows now from Theorem \ref{thm:graph-poly-trace-Psym}.
\end{proof}

Let us compute next the exact form of the first two moments. For the first moment, we have 
$$\E\Tr W^{\Gamma}= \frac{M}{2}J(\, \includegraphics[valign=c]{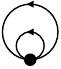}\, ;N) = \frac{MN(N+1)}{2} \sim \frac c 2 N^4.$$
Note that this computation is consistent with the result from Proposition \ref{prop:eigenval-eigenvec-average}. 

For the second moment ($p=2$), we have to sum over the cases $\alpha=\id$ and $\alpha=(12)$:
\begin{align*}\E \Tr\left((W^{\Gamma})^2\right) &=  \frac{M^2}{4}J(\, \includegraphics[valign=c]{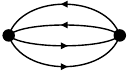}\, ;N)+ \frac{M}{4}J(\, \includegraphics[valign=c]{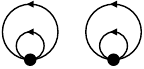}\, ;N)\\
&=\frac{M^2}{4}2N(N+1)+\frac{M}{4}N^2(N+1)^2 \sim \left(\frac{c^2}{2} + \frac c 4\right)N^6,
\end{align*}
where, for the case $\alpha = \id$, we have used Proposition \ref{prop:J-alpha=id} with $n=2$. 

To analyze larger moments, we need the following technical result. 

\begin{proposition}\label{prop:properties-J-G-alpha}
   For a given permutation $\alpha\in S_p$, consider the 2-in/2-out digraph $G_{\gamma\alpha, \gamma^{-1}\alpha}$ defined in Eq.~\eqref{eq:def-G-ga-gma}. Then, $G_{\gamma\alpha, \gamma^{-1}\alpha}$ has either one or two connected components. It has two connected components if and only if the permutation $\alpha$ is \emph{parity changing}: for all $i \in [p]$, $\alpha(i) \neq i (\mathrm{mod}\, 2)$. 
\end{proposition}
\begin{proof}
In the non-oriented graph $G_{\gamma \alpha, \gamma^{-1}\alpha}$, we have the following property: for any permutation $\sigma \in \langle \gamma\alpha, \gamma^{-1}\alpha\rangle$, and every $i \in [p]$, there is a path between the vertices $i$ and $\sigma(i)$. Note also that $\gamma^2\in \langle \gamma\alpha, \gamma^{-1}\alpha\rangle$. If $p$ is odd, $\gamma^2$ is a full cycle, hence $G_{\gamma \alpha, \gamma^{-1}\alpha}$ is connected. If $p$ is even, then $\gamma^2$ acts transitively on $\{2,4,\ldots, p\}$ and $\{1,3,\ldots, p-1\}$ separately. Thus $K(G_{\gamma \alpha, \gamma^{-1}\alpha})=1$ or $2$, depending on $\alpha$. Since $\gamma$ is a permutation that changes the parity, $K(G_{\gamma\alpha,\gamma^{-1}\alpha})=2$ if and only if $\alpha$ also changes the parity. This concludes the proof.
\end{proof}

We describe next the asymptotic behaviour of the larger moments of $W^\Gamma$. 
\begin{theorem}\label{thm:moments-WGamma-3-asympt}
For $p \geq 3$, in the limit $N \to \infty$, the moments of $W^{\Gamma}$ asymptotically behave as 
\begin{equation}\label{eq:asymptotics-partial}
    \E\Tr\left((W^{\Gamma})^p\right) = \frac{c^p}{2^p}N^{3p}+O(N^{3p-1}).
\end{equation}
\end{theorem}
\begin{proof}

We start from the sum over permutations from \eqref{eq:moments-W-Gamma}, and we distinguish three cases, as follows:
\begin{itemize}
    \item $\#\alpha = p$, i.e.~$\alpha=\id$. We know from Proposition \ref{prop:J-alpha=id} that $J(G_{\gamma,\gamma^{-1}};N) \sim N^p$, and thus the term $\alpha = \id$ behaves as $(c/2)^p N^{3p}$.
    
    \item $\#\alpha=p-1$, i.e.~$\alpha$ is a transposition. According to Lemma \ref{lem:UB-deg-J}, such a term is bounded by a polynomial whose leading term in $N$ has exponent 
    $$2\#\alpha + p+ K(G_{\gamma\alpha,\gamma^{-1}\alpha}) = 3p-2 + K(G_{\gamma\alpha,\gamma^{-1}\alpha}).$$
    Since $\alpha$ is a transposition and $p \geq 3$, $\alpha$ must have a fixed point, so it cannot be parity changing. Hence, using Proposition \ref{prop:properties-J-G-alpha}, $K(G_{\gamma\alpha,\gamma^{-1}\alpha}) =1$, proving that such terms are subdominant with respect to the previous case $\alpha = \id$.
    \item $\#\alpha\le p-2$. Reasoning as in the previous case, these terms are clearly subdominant, since $K(G_{\gamma\alpha,\gamma^{-1}\alpha}) \leq 2$.
\end{itemize}
In conclusion, all the terms with $\alpha \neq \id$ are subdominant with respect to the term $\alpha = \id$, finishing the proof.
\end{proof}

To finish this section, note that the asymptotic behaviour from \eqref{eq:asymptotics-partial}, corresponding to the moments $p \geq 3$, does not match the ones for $p=1,2$. This is a signature of the presence of an outlier, an eigenvalue on a larger scale, which is the phenomenon described in our main result, Theorem \ref{thm:main}.

\section{\texorpdfstring{$t$}{t}-channel random matrix and graphical representation}\label{sec:t-channel}

In this section we introduce the \emph{$t$-channel}\footnote{in reference to the $s$, $t$, $u$ channels of particle physics, as the diagrammatic representation is similar. In this spirit $W$ corresponds to the $s$-channel.} random matrix that we denote $W_{t}$. We do so as it is better suited for our aims in Sections \ref{sec:diagrammatics}, \ref{sec:bounds} and \ref{sec:tensor-eval}. We define $W_t$ component-wise by 
\begin{equation}
    (W_{t})_{ij,kl}:=(W)_{jl,ki}
\end{equation}
and we provide a diagrammatic representation of the above definition below
\begin{equation}\label{eq:diag-Wst}
    (W_{t})_{ij,kl}= \raisebox{-11mm}{\includegraphics[scale=0.6]{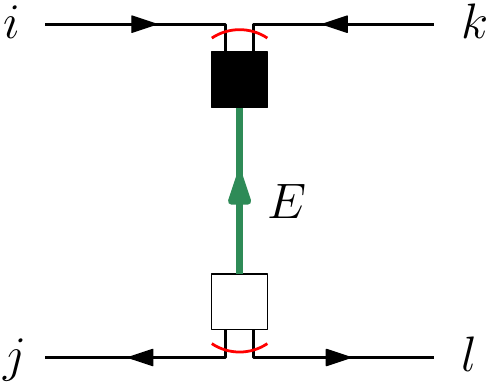}}.
\end{equation}
In the above diagrammatic representation the red arc represents the symmetrizer $P_s$. The black square represents the complex conjugate $\bar G$ of the Ginibre matrix $G$ of section \ref{sec:bosonic} while the white square represents $G$ itself. $\bar G$ is seen as a linear form on the Hilbert space, hence input vectors, which is why we display the corresponding tensor legs with ingoing edges. For similar reason we display tensor legs of $G$ with outgoing edges. The letter $E$ denotes the environmental Hilbert space that is traced out. Note also that we have the following equalities:
\begin{equation}\label{eq:diag-Wst*}
    (W_{t}^*)_{ij,kl}= \raisebox{-11mm}{\includegraphics[scale=0.6]{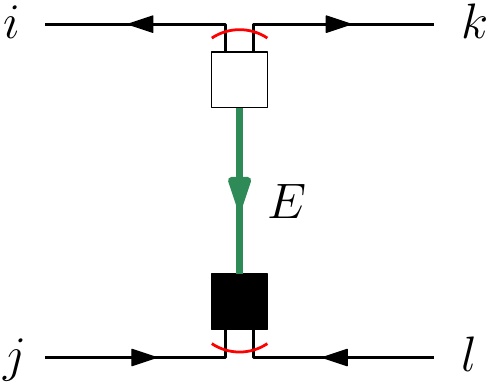}}
\end{equation}
We prefer to use black and white squares here to simplify the representations since we will use it intensively in the coming sections.
\\

We introduce the projector on the complement of $\mathbb C \ket \Omega$
\begin{equation}\label{eq:proj_complemet_omega_def}
    P_{\overline{\Omega}}=I-\lvert\Omega\rangle\langle \Omega \rvert.
\end{equation}
By the projector property we have $\Pcomp^2=\Pcomp$. We want to study the sequence of moments $\E(\Tr(Q^p))$ of the random matrix 
\begin{equation}
    Q:=\Pcomp W^{\Gamma} \Pcomp.
\end{equation}
In order to use graphical methods for tensor networks evaluation (see \cite{biamonte2017tensor} for an overview of tensor networks techniques and ideas), we will use the following relation 
\begin{equation}
    \Tr(Q^p)=\Tr((F \Pcomp W_{t} \Pcomp)^{p\textrm{ mod }2}(\Pcomp W_{t}^*\Pcomp W_{t} \Pcomp)^{\lfloor p/2 \rfloor}).
\end{equation}
This is easily shown by using the fact that $F^2 = I$, $F W_{t} F=W_{t}^*$ and $[\Pcomp,F]=0$. These moments have a graphical representation as ladder diagrams introduced in the next Section \ref{sec:diagrammatics}, and their Wick expansion produces terms that are indexed by tensor networks that are quotients of the ladder diagram by the action of Wick pairings. When there is no projectors $\Pcomp$, the tensor network evaluates to the circuit counting polynomial of the tensor network graph. However, due to the presence of the projector, there are two types of tensors appearing in the tensor network and this will give a different answer. This is the route we follow in Section \ref{sec:tensor-eval}. Note however that if we expand all projectors then we obtain, for each term of the expansion, a $2$-in/$2$-out digraph which represents a tensor network evaluating to the circuit counting polynomial of the graph. This is the method we use in this section and in the next Sections  \ref{sec:diagrammatics} and \ref{sec:bounds}.\\

We denote $P\models \{1,\ldots, p\}$ an \emph{interval partition} of $\{1,\ldots, p\}$. That is, $P$ is a set of subsets $P_i\subseteq \{1,\ldots, p\}$ of $P$ such that $\bigsqcup_i P_i=P$ and $P_i$ are sub-intervals of $\{1,\ldots, p\}$ seen cyclically. For instance,
\begin{itemize}
    \item $p=4$, $P=\{\{1,2\},\{3,4\}\},\ P_1=\{1,2\},\ P_2=\{3,4\}$ 
    \item $p=4$, $P=\{\{2,3\},\{4,1\}\},\ P_1=\{2,3\}, \ P_2=\{4,1\}$
    \item $p=4$, $P=\{\{2\},\{3\},\{4,1\}\},\ P_1=\{2\},\ P_2=\{3\}, \ P_3=\{4,1\}$. 
    \item $p=3$, $P=\{\{1,2,3\}\}, \ P_1=\{1,2,3\}$,
\end{itemize}
a non-example is given by the partition 
\begin{itemize}
    \item $p=5$, $P=\{\{1,2,4\},\{3,5\}\}$.
\end{itemize}
Assume we have a word $f\in \{0,1\}^p$. Assuming there exists $i$ such that $f(i)=1$, we associate to $f$ a subdivision of $P_f\models\{1,\ldots, p\}$ by putting bars between elements $i-1, i \in \{1,\ldots, p\}$ if and only if $f(i)=1$. $P_j$ is the interval defined as the set of elements between the $(j-1)^{\mathrm{th}}$ and $j^{\mathrm{th}}$ bar. 
\begin{figure}
    \centering
    \includegraphics[scale=0.7]{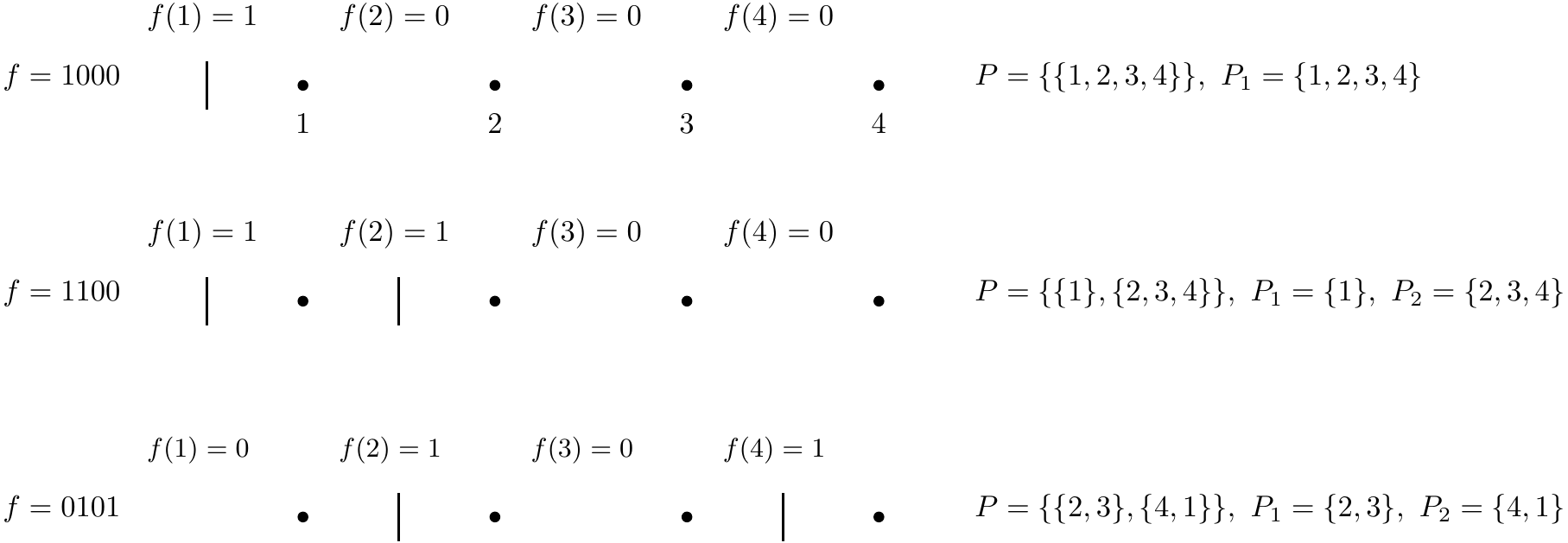}
    \caption{Examples of a interval subdivisions of $\{1,\ldots, 4\}$ associated to words $f\in \{0,1\}^4$ using the bars construction.}
    \label{fig:bars-to-intervals}
\end{figure}
See Fig. \ref{fig:bars-to-intervals} for a pictorial description.  \\

Expanding the projectors $\Pcomp$ appearing in the definition of $Q$, we have 
\begin{multline}\label{eq:Pcomp-expansion}
    \Tr(Q^p)= \Tr((F W_{t})^{p \textrm{ mod }2}(W_{t}^*W_{t})^{\lfloor p/2\rfloor})\\
    +\sum_{\substack{f\in \{0,1\}^p \\ \exists i, f(i)=1}}\left(\frac{-1}{N}\right)^{|P_f|} \prod_{P_{f,i}\in P}N\langle \Omega \vert (W_{t})^{|P_{f,i}| \textrm{ mod }2}(W_{t}^*W_{t})^{\left\lfloor |P_{f,i}|/2\right\rfloor} \lvert \Omega \rangle
\end{multline}
This expansion forms the basis for the expansion of the moments of $Q$ as sums of circuit counting polynomials of $2$-in/$2$-out digraphs. The term of the form 
\begin{equation}\label{eq:exp-moments-Q}
    N\langle \Omega \vert (W_{t})^{|P_{f,i}| \textrm{ mod }2}(W_{t}^*W_{t})^{\left\lfloor |P_{f,i}|/2\right\rfloor} \lvert \Omega \rangle
\end{equation}
can be represented diagrammatically, and we will use this diagrammatic representation to understand the Wick expansion of the moments of $Q$.\\
\begin{remark}\label{rem:correspondance-W-Partial}
Note also that the first term of the right hand side of the equation \eqref{eq:exp-moments-Q} above is $\Tr((W^{\Gamma})^p)$. We use this remark later to deduce the diagrammatics of $\E\left( \Tr((W^{\Gamma})^p) \right)$.
\end{remark}

\section{Diagrammatics for the moments of \texorpdfstring{$Q$}{Q} and \texorpdfstring{$W_{t}$}{Wt}}\label{sec:diagrammatics} 

We start by considering the diagram representation of expression involving $W_{t}$ in the expansion \eqref{eq:Pcomp-expansion}. They are easily obtained by stacking building blocks of equations \eqref{eq:diag-Wst},\eqref{eq:diag-Wst*}. We have 
\begin{multline}\label{eq:trace-to-ladder-even}
   \forall p \in 2\mathbb{N}, \ \Tr((F W_{t})^{p \textrm{ mod }2})(W_{t}^*W_{t})^{\lfloor p/2\rfloor} = \Tr((W_{t}^*W_{t})^{p/2})\\
   = \raisebox{-11mm}{\includegraphics[scale=0.5]{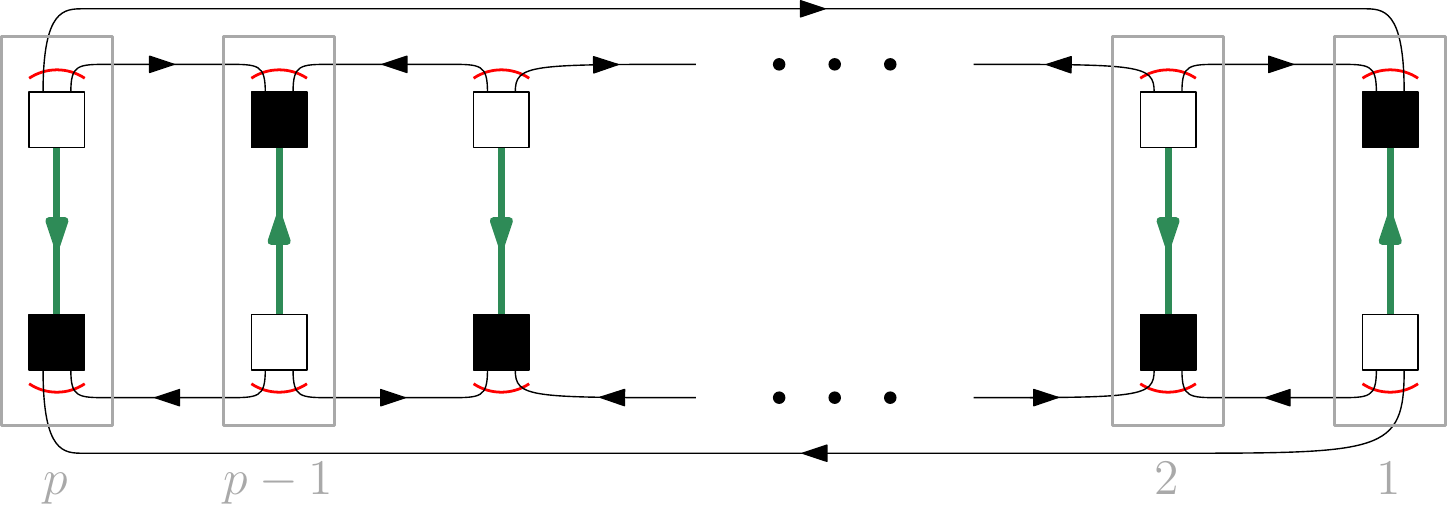}}\
\end{multline}
If $p\in 2\mathbb{N}+1$ is odd, we have the slightly twisted representation
\begin{equation}\label{eq:trace-to-ladder-odd}
    \Tr((F W_{t})(W_{t}^*W_{t})^{\lfloor p/2\rfloor}) \ = \ \raisebox{-12mm}{\includegraphics[scale=0.5]{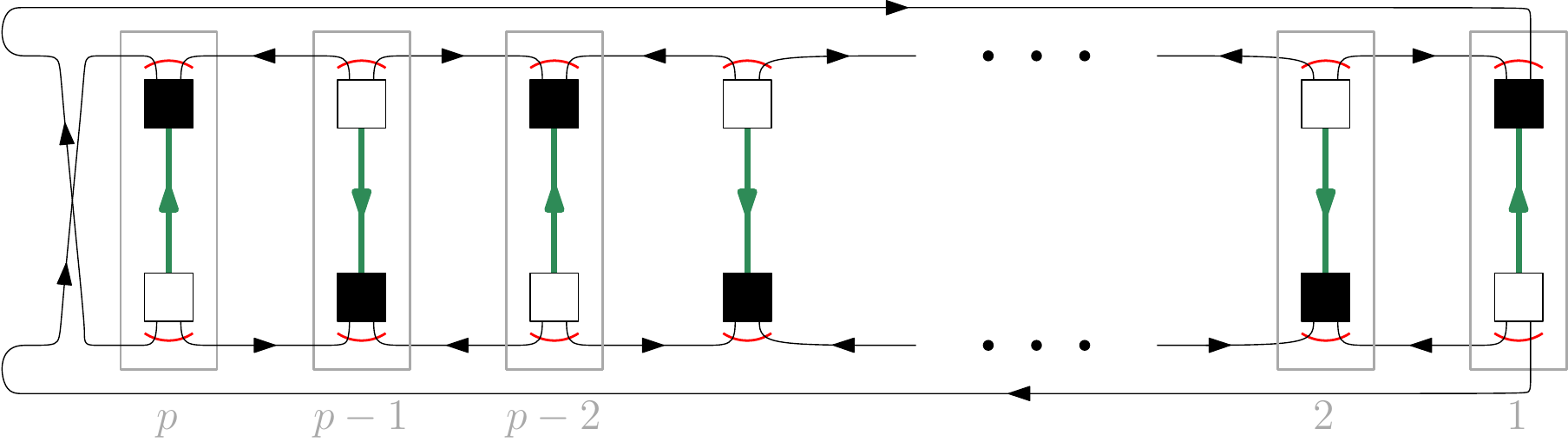}}
\end{equation}\\
The gray boxes highlight the copies of $W_{t}$ (resp. $W_{t}^*$) which appear in the product in the trace. We number these copies as shown on the graphical representation. The black and white squares inherit the numbering of the gray box they lie in. This diagrammatic allows us to describe the Wick pairings as permutations $\alpha \in S_p$. Indeed, since $G$ is normally distributed, with vanishing pseudo-variance, Wick pairings only pair instances of $G$ with instances of $\bar G$. We represent such a pairing by adding oriented dotted edges from black squares to white ones in diagrams of the type appearing in equations \eqref{eq:trace-to-ladder-even} and \eqref{eq:trace-to-ladder-odd}. Each Wick pairing can be indexed by a permutation in $\alpha \in S_p$. Indeed the black box number $i$ is paired with the white box $j$ if and only if $\alpha(i)=j$. This translates diagrammatically as a dotted edge between black square $i$ and white square $j$, oriented from black to white. See examples on Fig. \ref{fig:pairing-example-trace-to-ladder}.
\begin{figure}
    \centering
    \includegraphics[scale=0.6]{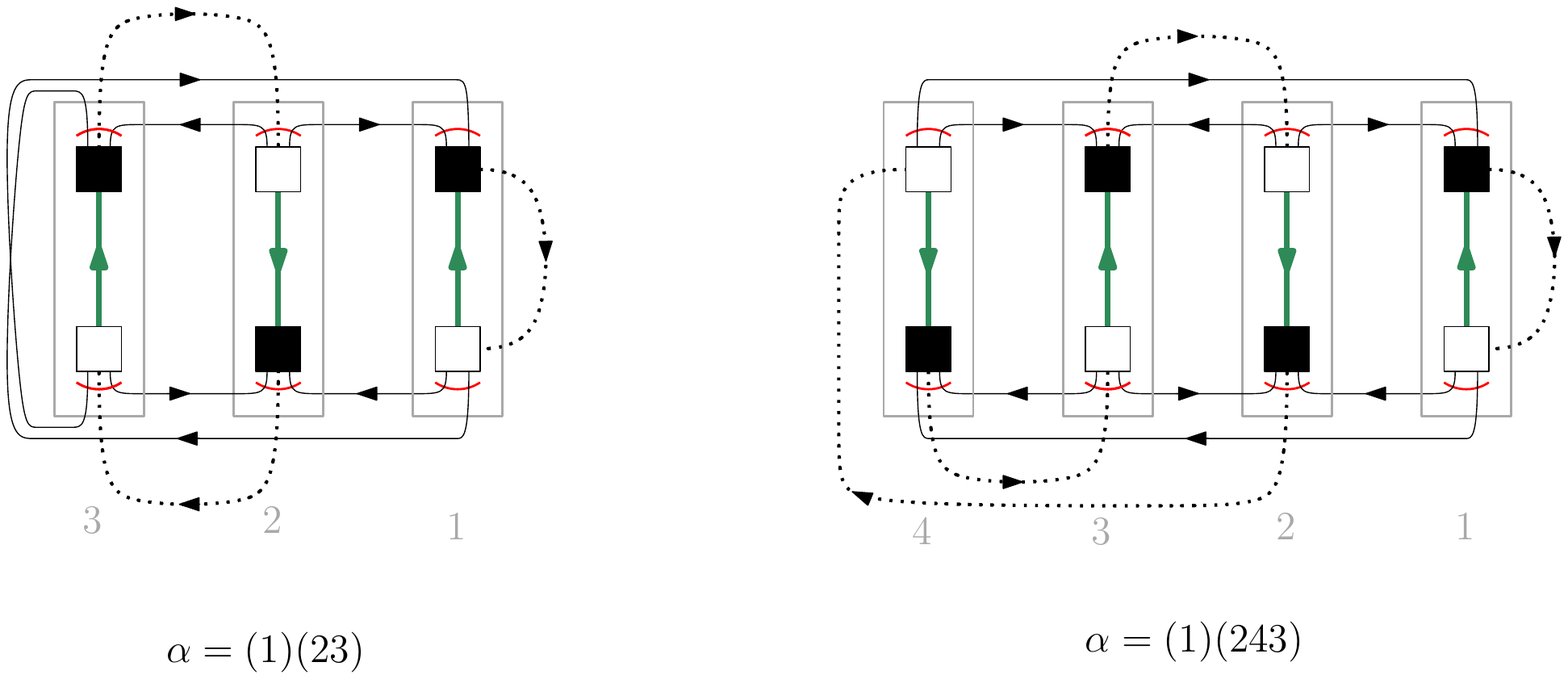}
    \caption{Two examples of Wick pairings indexed by permutations $\alpha$. The dotted lines allows to track how the random states are paired together.}
    \label{fig:pairing-example-trace-to-ladder}
\end{figure}\\

Now we realize that each such pairing tells us about the composition of symmetrizers (denoted as red circle arcs), whose inputs and outputs are kept track of using the orientations of the black edges in the diagrams. In fact, these Wick pairings tell us how to match the indices of these symmetrizers together. In particular, note that if $\alpha(i)=j$, then the symmetrizer associated to the black box $i$ is composed (in the usual operator sense) with the symmetrizer associated with the white box $j$. Since the symmetrizer is a projector, we are left with a unique symmetrizer for each dotted edge whose inputs are the black edges ending on the black box $i$ and whose outputs are the black edges exiting the white box $j$. Since a symmetrizer symmetrizes over both the inputs and the outputs it is not necessary to keep track of the difference between a first input (resp. output) and a second input (resp. output). Thus, the corresponding contraction of symmetrizers is indexed by a directed graph (digraph) whose vertices are $2$-in/$2$-out and represent two\footnote{the factor of $2$ is conventional. Its purpose is to recover the definition of the $J$ polynomial later.} times $P_{s}$. Such a digraph can be seen as a tensor network for the tensor $P_s$ The value of the tensor network is just the evaluation of the corresponding contraction. We explain in details this evaluation combinatorially. Since each vertex represents twice a symmetrizer $P_{s}=\frac12(I+F)$, we can evaluate the contraction by deciding for a \emph{state}\footnote{we use here the terminology that is familiar to the graph polynomial literature. Note though that vertex states should not be confused with quantum states. We hope that the context will be clear enough to avoid confusion.}, that is the assignation of either the identity $I:\C^N\otimes \C^N \rightarrow \C^N\otimes \C^N$ or the flip $F:\C^N\otimes \C^N \rightarrow \C^N\otimes \C^N$, at each vertex. The possible states $\{S_1, S_2\}$ are 
\begin{equation}
    \raisebox{-6mm}{\includegraphics[scale=0.7]{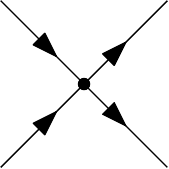}}\longrightarrow \raisebox{-12mm}{\includegraphics[scale=0.7]{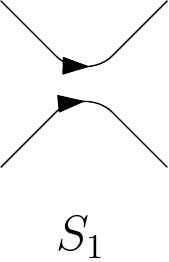}}, \quad \raisebox{-12mm}{\includegraphics[scale=0.7]{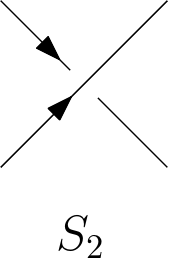}}.
\end{equation}
For each state assignation to all the vertices of the graph, we obtain a collection of cycles, called a cycle cover of the graph. Each of these cycles are weighted by $N$. Hence the weight of a particular state assignation is $N$ raised to the number of resulting cycles.  Summing over all possible state assignation will result in the circuit counting polynomial of the graph. Hence the weight of the tensor network we constructed is just the circuit counting polynomial of the underlying graph times $\left( \frac12 \right)^p$, with $p$ the number of vertices. This last factor just takes into account the normalization factor $\frac12$ of the symmetrizer.  \\

The digraph resulting from the Wick pairing $\alpha$ can be indexed by the data of two permutations that also allows us to construct directly the adjacency matrix of the digraph. \

We describe this permutation representation of $2$-in/$2$-out digraphs obtained from a Wick pairing. These digraphs are indexed by $\sigma_1=\gamma\alpha, \ \sigma_2=\gamma^{-1}\alpha$ with $\gamma=(1,2,\ldots, p)$. Indeed, a vertex labeled $i$ corresponds, by definition, to a pair $(i,\alpha(i))$. The edges entering this vertex are the edges entering the black square of the pair $i$ in the ladder graph while the edges exiting this vertex are the edges exiting the white square of the gray pair $\alpha(i)$. Note that the edges exiting from vertex $i$ of the digraph, are adjacent to the black boxes of the gray pairs $\gamma(\alpha(i))$ and $\gamma^{-1}(\alpha(i))$, that is in the digraph they are adjacent to the vertices $\gamma(\alpha(i))$,  $\gamma^{-1}(\alpha(i)$ oriented as $i\rightarrow \gamma(\alpha(i))$ and $i\rightarrow \gamma^{-1}(\alpha(i))$. Hence, $\sigma_1=\gamma\alpha, \ \sigma_2=\gamma^{-1}\alpha$. See the local construction of the vertices on Fig. \ref{fig:local-construct}.\\

\begin{figure}
    \centering
    \includegraphics[scale=0.9]{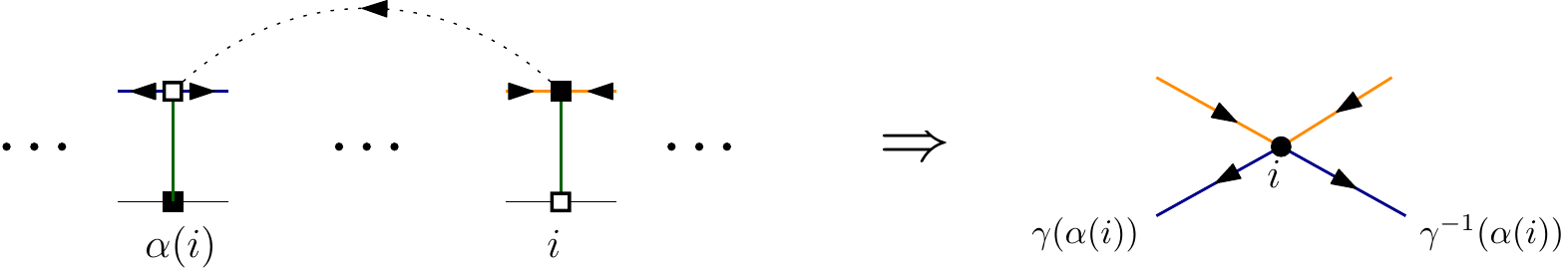}
    \caption{On this figure we show how the vertices of the digraph are constructed from the pairing $\alpha$ by identifying the black box $i$ with the white box $\alpha(i)$. We forget the environment (green) lines in the digraph. We highlighted the edges adjacent to the black and white boxes to be identified in orange and blue with their orientations so that it is easier to spot them in the digraph.}
    \label{fig:local-construct}
\end{figure}

We follow a similar train of thoughts for the terms in the sum 
\begin{equation}
    \sum_{\substack{f\in \{0,1\}^p \\ \exists i, f(i)=1}}\left(\frac{-1}{N}\right)^{|P_f|} \prod_{P_{f,i}\in P}N\langle \Omega \vert (W_{t})^{|P_{f,i}| \textrm{ mod }2}(W_{t}^*W_{t})^{\left\lfloor |P_{f,i}|/2\right\rfloor} \lvert \Omega \rangle
\end{equation}
from equation \eqref{eq:Pcomp-expansion}. Indeed, the product $(W_{t})^{|P_{f,i}| \textrm{ mod }2}(W_{t}^*W_{t})^{\left\lfloor |P_{f,i}|/2\right\rfloor}$ is also represented diagrammatically by stacking building blocks of equations \eqref{eq:diag-Wst} and \eqref{eq:diag-Wst*}. The only difference being that the vector $\lvert \Omega \rangle$ has the following representation 
\begin{equation}
    \lvert \Omega \rangle =\frac1{\sqrt{N}} \ \raisebox{-3mm}{\includegraphics[scale=0.6]{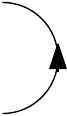}}
\end{equation}
Hence, we have the diagrammatic below
\begin{equation}
    N \langle \Omega \rvert(W_{t})^{|P_{f,i}| \textrm{ mod }2}(W_{t}^*W_{t})^{\left\lfloor |P_{f,i}|/2\right\rfloor}\lvert\Omega\rangle = \ \raisebox{-12mm}{\includegraphics[scale=0.48]{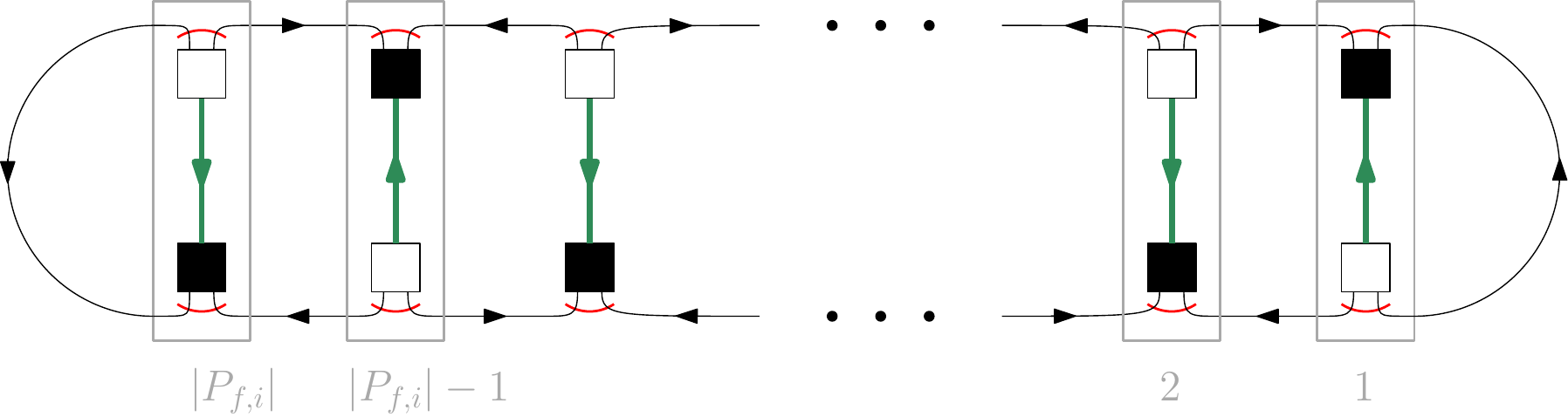}}
\end{equation}
We introduce permutations that we use to give a description of the digraphs appearing in the computation of the above sum. We start with the following permutations on $l$ elements, whose definition differs in the odd and even cases,
\begin{align}
    \gamma_{a,l}&=\begin{cases}
    (1)(23)(45)\ldots(2k,2k+1), \ \textrm{if }l=2k+1 \textrm{ is odd} \\
    (1)(23)(45)\ldots(2k-2,2k-1)(2k), \ \textrm{if }l=2k \textrm{ is even}. 
    \end{cases}\\
    \gamma_{b,l}&=\begin{cases}
    (12)(34)\ldots(2k-1,2k)(2k+1), \ \textrm{if }l=2k+1 \textrm{ is odd} \\
    (12)(34)\ldots(2k-1,2k), \ \textrm{if }l=2k \textrm{ is even}
    \end{cases}
\end{align}
and to each $P_i\in P \models \{1,\ldots, p\}$ we associate the bijection $C_{P_i}:P_i\rightarrow \{1,\ldots, \lvert P_i\rvert\}$ defined by
\begin{equation}
    C_{P_i}(P_i(q))=q
\end{equation}
where $P_i(q)$ is the $q^{\textrm{th}}$ element of $P_i$. Then we define for each $P_i$ the permutation $\gamma_{a,P_i}=C_{P_i}^{-1}\gamma_{a,\lvert P_i\rvert}C_{P_i}$ and $\gamma_{b,P_i}=C_{P_i}^{-1}\gamma_{b,\lvert P_i\rvert}C_{P_i}$. Finally, we construct yet another two permutations, $\Gamma_{a,P}, \, \Gamma_{b,P} := \prod_{i}\gamma_{a,P_i}, \, \prod_{i}\gamma_{b,P_i}$. We have several figures to illustrate the relations with the diagrams. See Fig. \ref{fig:overlaps-gamma-ab} for $\gamma_{a,\lvert P_i\rvert}$ and $\gamma_{b,\lvert P_i\lvert}$. See also Fig. \ref{fig:example_GammaP} for an example of permutations $\Gamma_{a,P}, \Gamma_{b,p}$.\\

\begin{figure}
    \centering
    \includegraphics[scale=0.6]{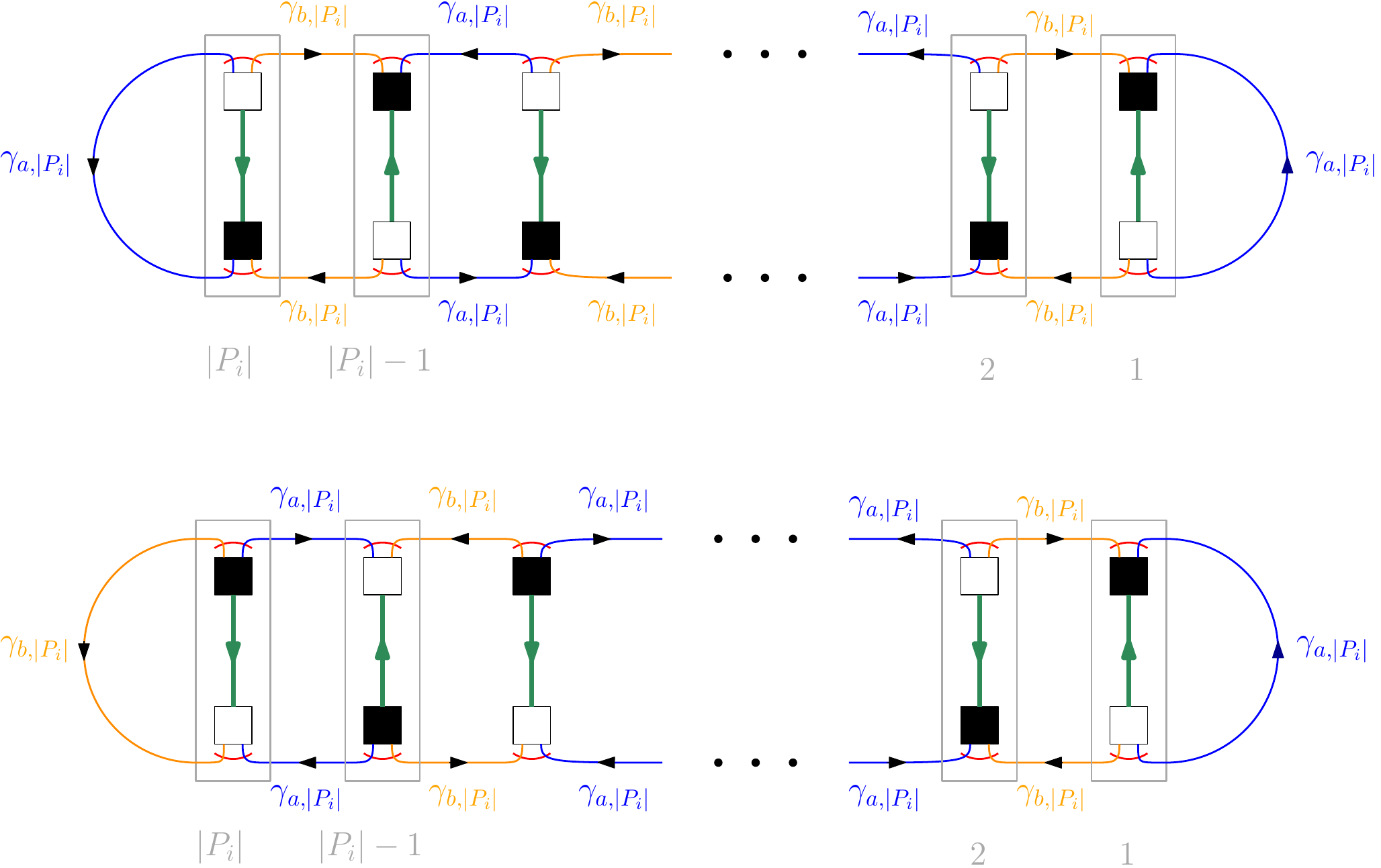}
    \caption{Illustration of the relation between the diagrams and the permutations $\gamma_{a,|P_i|}$ and $\gamma_{b,|P_i|}$. On the topmost picture: the even $|P_i|$ case. On the bottom-most picture: the odd $|P_i|$ case.}
    \label{fig:overlaps-gamma-ab}
\end{figure}

Then we have  as a consequence of Wick theorem that: 
\begin{equation}\label{eq:Wick-prod-overlaps}
    \E\left( \prod_{P_{f,i}\in P}N\langle \Omega \vert (W_{t})^{|P_{f,i}| \textrm{ mod }2}(W_{t}^*W_{t})^{\left\lfloor |P_{f,i}|/2\right\rfloor} \lvert \Omega \rangle\right) = \sum_{\alpha \in S_p}N^{2\# \alpha}\frac{c^{\#\alpha}}{2^p} J(G_{\Gamma_{a,P}\alpha,\Gamma_{b,P}\alpha},N)
\end{equation}
where we recall that $\sum_i \lvert P_{f,i}\rvert =p$.\\

In the same way than in the previous case, the Wick theorem identifies black box $i$ with white box $j$ for any pairing such that $\alpha(i)=j$. For the same reason than before, it leads to a $2$-in/$2$-out digraph. However, now the edges exiting a vertex $i$ of the digraph are the edges exiting the white box $\alpha(i)$ while these edges are adjacent to the black boxes $\Gamma_{a,P}(\alpha(i)), \Gamma_{b,P}(\alpha(i))$, which correspond to the vertices $\Gamma_{a,P}(\alpha(i))$ and $\Gamma_{b,P}(\alpha(i))$ in the digraph. This is sufficient to conclude that we have the relation \eqref{eq:Wick-prod-overlaps}. \\

We summarize the content of the above discussion by the proposition below:
\begin{proposition}\label{prop:expanded-Wick-thm}
   As a consequence of Wick theorem we have,
   \begin{equation}
       \E\left(\Tr((F W_{t}^{p \textrm{ mod }2})(W_{t}^*W_{t})^{\lfloor p/2\rfloor})\right) = \sum_{\alpha\in S_p}N^{2\# \alpha}\frac{c^{\#\alpha}}{2^p}J(G_{\gamma\alpha, \gamma^{-1}\alpha};N)
   \end{equation}
   and
   \begin{multline}
       \sum_{\substack{f\in \{0,1\}^p \\ \exists i, f(i)=1}}\left(\frac{-1}{N}\right)^{|P_f|} \prod_{P_{f,i}\in P}N\langle \Omega \vert (W_{t})^{|P_{f,i}| \textrm{ mod }2}(W_{t}^*W_{t})^{\left\lfloor |P_{f,i}|/2\right\rfloor} \lvert \Omega \rangle =\\
       \sum_{\substack{f\in \{0,1\}^p \\ \exists i, f(i)=1}}\sum_{\alpha \in S_p} \left(-1\right)^{|P_f|} N^{2\#\alpha - \lvert P_f\rvert } \frac{c^{\#\alpha}}{2^p}J(G_{\Gamma_{a,P_f}\alpha,\Gamma_{b,P_f}\alpha };N). 
   \end{multline}
\end{proposition}
 Reformulating the above proposition and using the Remark \ref{rem:correspondance-W-Partial}, we can recover the moments of the random matrix $W^\Gamma$ from Proposition \ref{prop:moments-W-Gamma} 
\begin{equation}
    \E\left(\Tr\left((W^{\Gamma})^p\right)\right) = \sum_{\alpha\in S_p}N^{2\# \alpha}\frac{c^{\#\alpha}}{2^p}J(G_{\gamma\alpha, \gamma^{-1}\alpha};N).
\end{equation}
We use the short-hand notation, for $\alpha\in S_p$ 
\begin{equation}
    w(\alpha,\gamma):=N^{2\# \alpha}\frac{c^{\#\alpha}}{2^p}J(G_{\gamma\alpha,\gamma^{-1}\alpha}).
\end{equation}

\begin{figure}
    \centering
    \includegraphics[scale=0.55]{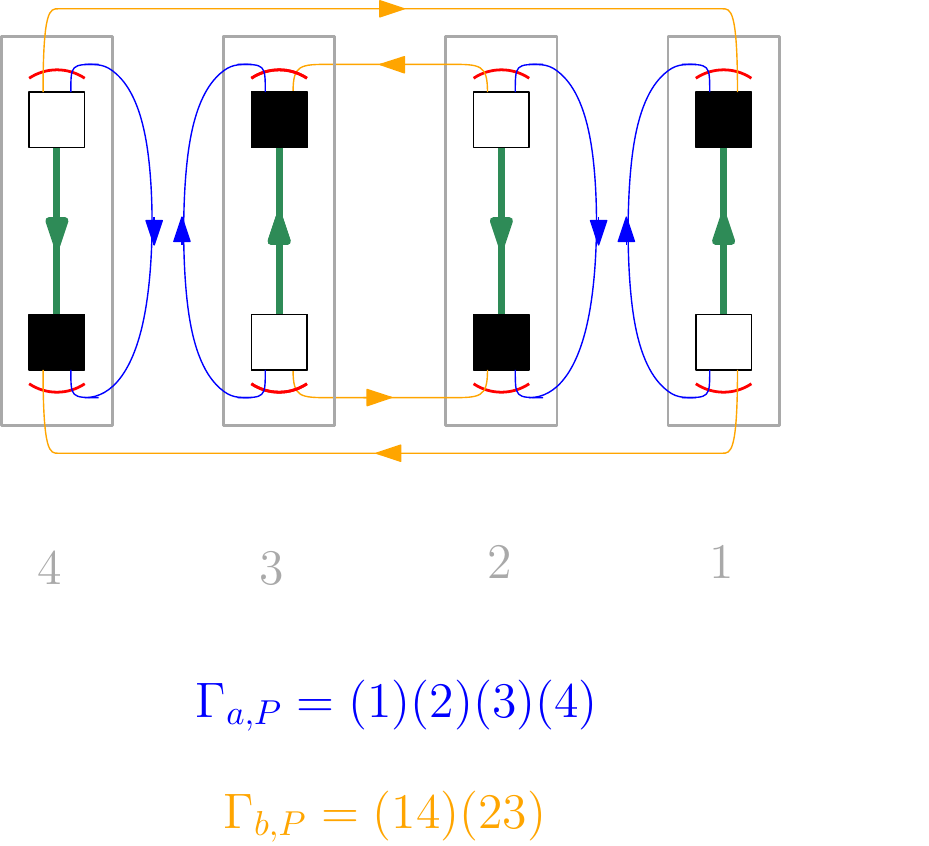}
    \caption{Example of permutations $\Gamma_{a,P}$, $\Gamma_{b,P}$ for $P=\{\{2,3\},\{1,4\}\}$.}
    \label{fig:example_GammaP}
\end{figure}

\section{Circuit polynomial techniques for \texorpdfstring{$\E\Tr(Q^p)$}{E Tr Qp}}\label{sec:bounds}

In this section we prove a sequence of technical bounds and combinatorial properties that we use to determine the set of permutations $\alpha$ contributing to the asymptotics of the random matrices $W^\Gamma$ and $Q$. \\

We have the following bound
\begin{proposition}\label{prop:bound_prop}
   Let $\alpha\in S_p$, then there exists $R_{\alpha,f}>0$, that can depend on $\alpha$ and on $f$ but not on $N$
   \begin{equation}
       N^{\#2\alpha-\lvert P_f \rvert }J(G_{\Gamma_{a,P_f}\alpha,\Gamma_{b,P_f}\alpha }; N)\le R_{\alpha, f}N^{2\#\alpha - \lvert P_f\rvert +p + K(G_{\Gamma_{a,P_f}\alpha,\Gamma_{b,P_f}\alpha })},
   \end{equation}
   with $K(G_{\Gamma_{a,P_f}\alpha,\Gamma_{b,P_f}\alpha })$ the number of connected components of the corresponding digraph. Moreover, we have: $$K(G_{\Gamma_{a,P_f}\alpha,\Gamma_{b,P_f}\alpha })-\lvert P_f\rvert \le 0.$$
\end{proposition}
\begin{proof}
We now come to the second bound in the case of $2$-in/$2$-out digraphs $G_{\Gamma_{a,P_f}\alpha,\Gamma_{b,P_f}\alpha }$. The main difference lies in their number of connected component. Indeed, $\Gamma_{a,P_f}\Gamma_{b,P_f}\in \langle \Gamma_{a,P_f}\alpha, \Gamma_{b,P_f}\alpha\rangle$ acts transitively on the sets $P_{f,i}$ separately, but is not transitive on $\{1,\ldots,p\}$. Thus the number of connected components is bounded from above by $\lvert P_f \rvert$. Hence we have 
\begin{equation}
    K(G_{\Gamma_{a,P_f}\alpha,\Gamma_{b,P_f}\alpha })-\lvert P_f\rvert \le 0.
\end{equation}
Again, using the relation between the circuit polynomial and the interlace polynomial on each connected component of $G_{\Gamma_{a,P_f}\alpha,\Gamma_{b,P_f}\alpha }$, we obtain 
\begin{equation}
    \textrm{deg }J(G_{\Gamma_{a,P_f}\alpha,\Gamma_{b,P_f}\alpha }; N)\le p+K(G_{\Gamma_{a,P_f}\alpha,\Gamma_{b,P_f}\alpha }).
\end{equation}
\end{proof}

We also show the two next technical Lemma \ref{lem:maximal-scaling-cycle-type} and \ref{lem:crossing-length2}. They prove useful in the last section, to determine the limiting empirical distribution of eigenvalues of $Q$,
\begin{lemma}\label{lem:maximal-scaling-cycle-type}
Let $\alpha \in S_p$ such that $2\#\alpha +p + K(G_{\gamma\alpha,\gamma^{-1}\alpha})= 2p+e$ with $e$ a positive integer. Then, if $K(G_{\gamma\alpha,\gamma^{-1}\alpha})=1$
\begin{equation}
    e=(\#_1\alpha+1)+\sum_{i\ge 3}(2-i)\#_i\alpha\le \#_1\alpha+1,
\end{equation}
where $\#_i\alpha$ is the number of cycles of length $i$ of $\alpha$. If  $K(G_{\gamma\alpha,\gamma^{-1}\alpha})=2$,
\begin{equation}
    e=2+\sum_{k\ge 2}2(1-k)C_{2k}\le 2.
\end{equation}
Similarly, let $e'$ be a positive integer such that $2\#\alpha +p + K(G_{\Gamma_{a,P}\alpha,\Gamma_{b,P}\alpha})=2p+e'$. Then,
\begin{equation}
    e'=\#_1\alpha+\sum_{i\ge3}(2-i)\#_i\alpha+\left(K(G_{\Gamma_{a,P}\alpha,\Gamma_{b,P}\alpha})-\lvert P \rvert \right).
\end{equation}
In particular, if $\alpha$ is not allowed to have fixed points, then $e'\le 0$. If we have equality $e'=0$, then $\#_i\alpha=0$ for all $ i \ge 3$.
\end{lemma}
\begin{proof}
We first focus on $e$. We denote $\#_i\alpha$ the number of cycles of length $i$ of $\alpha$. We have the following relations
\begin{equation}
  p=\sum_{i\ge 1} i \#_i\alpha \quad \textrm{and} \quad \#\alpha =\sum_{i\ge 1}\#_i\alpha.
\end{equation}
From those we deduce
\begin{equation}
    e= \sum_{i\ge 1}(2-i)\#_i\alpha+K(G_{\gamma\alpha,\gamma^{-1}\alpha}).
\end{equation}
We have two cases, first assume that $K(G_{\gamma\alpha,\gamma^{-1}\alpha})=1$. Then, a consequence of proposition \ref{prop:properties-J-G-alpha} is that $\alpha$ is allowed to have fixed points, thus
\begin{equation}
    e=(\#_1\alpha+1)+\sum_{i\ge 3}(2-i)\#_i\alpha \le \#_1\alpha+1.
\end{equation}
Now, assume that $K(G_{\gamma\alpha,\gamma^{-1}\alpha})=2$. A consequence of proposition \ref{prop:properties-J-G-alpha} is that $\alpha$ is not allowed to have fixed point (as those are not parity changing). More generally, $\alpha$ is not allowed to have cycles of odd length. Therefore,
\begin{equation}
    e=2+\sum_{k\ge 2}2(1-k)\#_{2k}\alpha\le 2.
\end{equation}
The above inequality is saturated when $\alpha$ has only cycles of length $2$ as any longer cycles will contribute negatively to the sum of the right hand side. This proves the two first statements of the lemma.\\
We now come to $e'$. We have
\begin{equation}
    e'=\#_1\alpha+\sum_{i\ge 3}(2-i)\#_i\alpha+\left(K(G_{\Gamma_{a,P}\alpha, \Gamma_{b,P}\alpha})-\lvert P \rvert \right),
\end{equation}
where the second and last term of the right hand side is bounded from above by $0$. Hence, $e'\le \#_1\alpha$. In particular, if $\alpha$ does not have fixed points, $e'\le 0$ with equality if $\#_i\alpha=0 \ \forall i \ge 3$ and $\alpha(P_i)=P_i, \ \forall P_i \in P$. 
\end{proof}

\noindent As we understand from the above proposition, fixed points and cycles of length $2$ of $\alpha$ play an important role. Indeed, the scaling of a given Wick pairing $\alpha$ can exceed $N^{2p+2}$ if and only if there are enough fixed points, and not too many cycles of length greater or equal to $3$. The lemma below is an important step to the proof that we discuss fully in the next sections. 

\begin{lemma}\label{lem:crossing-length2}
Assume $\alpha\in S_p$ has only cycles of length $2$ and is parity changing (\textit{i.e.} $\alpha(i) \textrm{ mod }2 \neq i \textrm{ mod }2$). Assume also that there exists two cycles of $\alpha$ of the form $(a,i+1)$ and $(i,b)$ for $a<i<i+1<b$. Then, $\textrm{deg }J(G_{\gamma\alpha,\gamma^{-1}\alpha};N)\le p$. 
\end{lemma}

We use the parity changing assumption together with the constraint that $\alpha$ has only length $2$ cycles only to simplify the proof and avoid having to consider several cases. Indeed, a stronger form of this result can easily be proven from similar considerations. However, the above form of the lemma will be sufficient for our purposes.

\begin{proof}
The parity changing assumption ensures that $G_{\gamma\alpha,\gamma^{-1}\alpha}$ has $2$ connected components. We note that these two connected components must be isomorphic graphs. Indeed, we can exhibit an isomorphism. For each cycle of the form $(k,l)$ of $\alpha$, the isomorphism $\phi$ is set to $\phi(k)=l$. Hence, we can focus on one connected component only and bound the degree of the circuit counting polynomial of this connected component. We use Fig.~\ref{fig:crossing-deg-bound} to show the structure of $G_{\gamma\alpha,\gamma^{-1}\alpha}$ around vertices $a,b,i,i+1$. From this figure we note that neither vertex $a$ nor vertex $i$ are cut vertices. Moreover, since the figure we show is minimally connected (gray parts inside each connected components could be further connected if $\alpha$ had more crossings), this means vertices $a$ and $i$ cannot become cut-vertices by changing $\alpha$ while keeping the cycles $(a,i+1)$ and $(i,b)$ fixed.\\

Denote $G^{(a)}_{\gamma\alpha,\gamma^{-1}\alpha}$ the connected component containing vertex $a$. We focus on this connected component for now. By Lemma \ref{lem:max-J-only-cut}, we have $\textrm{deg }J(G_{\gamma\alpha,\gamma^{-1}\alpha}^{(a)};N)\le p/2$. Indeed, since for instance $a$ is not a cut-vertex we know that once we reduced vertex $a$ by summing over its two possible states using the skein relations of Eq.~\eqref{eq:skein-relation-J} on $a$ we are left with two graphs $G'$ and $G''$ such that 
\begin{equation}
    J(G^{(a)}_{\gamma\alpha,\gamma\alpha^{-1}};N)=J(G';N)+J(G''; N).
\end{equation}
Both $G'$ and $G''$ have a maximum of $p/2-1$ vertices. Therefore, from Lemma \ref{lem:UB-deg-J} 
\begin{equation}
    \textrm{deg }J(G^{(a)}_{\gamma\alpha,\gamma\alpha^{-1}};N)=\textrm{max}(\textrm{deg }J(G';N), \textrm{deg }J(G''; N)) \le p/2.
\end{equation}
By the isomorphism argument, we have the same bound on the degree of $J(G^{(b)}_{\gamma\alpha, \gamma^{-1},\alpha};N)$, hence
\begin{equation}
    \textrm{deg } J(G_{\gamma\alpha, \gamma^{-1}\alpha})\le p.
\end{equation}

\end{proof}

In the next section, we show that the contribution of fixed points is smaller than the previous bounds would let us believe. In order to study fixed points and cycles of length $2$ we introduce more general tensor networks.\\
\begin{figure}
    \centering
    \includegraphics[scale=0.6]{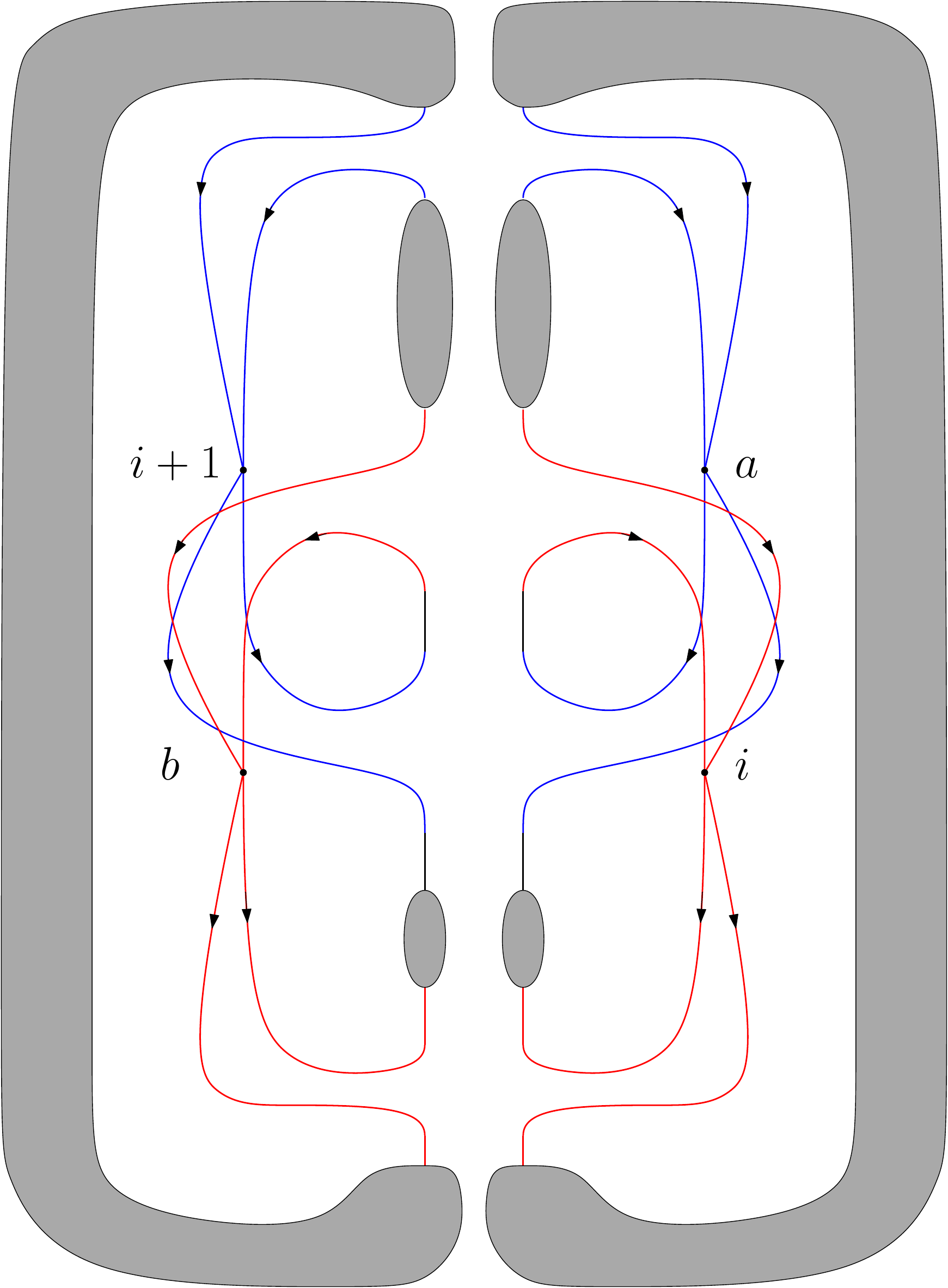}
    \caption{Local structure of the graph $G_{\gamma\alpha,\gamma^{-1}\alpha}$ for $\alpha$ satisfying the constraints of Lemma \ref{lem:crossing-length2}. The axe of symmetry corresponds to the ladder axe of symmetry. The figure shows the smallest connexity case. Indeed, the gray part of each connected component could be further connected if $\alpha$ have more crossings. For readability we color the edges differently depending on which vertices they are adjacent to.  The color change in the middle if the edges are shared between vertices $a$ and $i$ or between vertices $b$ and $i+1$.}
    \label{fig:crossing-deg-bound}
\end{figure}

\section{Tensor network evaluation of \texorpdfstring{$\E\Tr(Q^p)$}{E Tr Qp}}\label{sec:tensor-eval}

In this section, we express the moments of the random matrix $Q$ as a sum over evaluations of tensor networks made of the tensors $P_{s}$ and $P_{\overline \Omega}$. In particular, we do not expand the projectors $P_{\overline \Omega}$ appearing in the expression of $Q$. In that case, for each Wick pairing $\alpha$ we obtain a contraction of two types of tensors, $P_{s}:\bC^{N}\otimes \bC^N\rightarrow \bC^{N}\otimes \bC^N$ as before and $P_{\overline{\Omega}}:\bC^{N}\otimes \bC^N\rightarrow \bC^{N}\otimes \bC^N$ which is the projector on the complement of $\textrm{span}\{\lvert \Omega \rangle\}$ introduced earlier in Eq.~\eqref{eq:proj_complemet_omega_def}.

We work here with a graphical representation that is very similar to the ones introduced in previous paragraphs. Here, we represent $\Tr(Q^p)$ directly, without expanding. We do so by stacking the building blocks introduced in \eqref{eq:diag-Wst}, \eqref{eq:diag-Wst*} together with an additional building block representing $P_{\overline \Omega}$
\begin{equation}
    P_{\overline \Omega} = I- \lvert \Omega \rangle \langle \Omega \rvert = \raisebox{-6mm}{\includegraphics[scale=0.8]{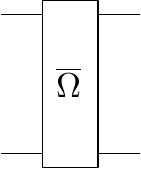}}
\end{equation}
These building blocks are stacked in between  \eqref{eq:diag-Wst}, \eqref{eq:diag-Wst*} blocks to represent the operator $P_{\overline \Omega}$. The diagram thus obtained is a tensor network representing $\Tr(Q^n)$
\begin{equation}
    \Tr(Q^n)=\raisebox{-12mm}{\includegraphics[scale=0.50]{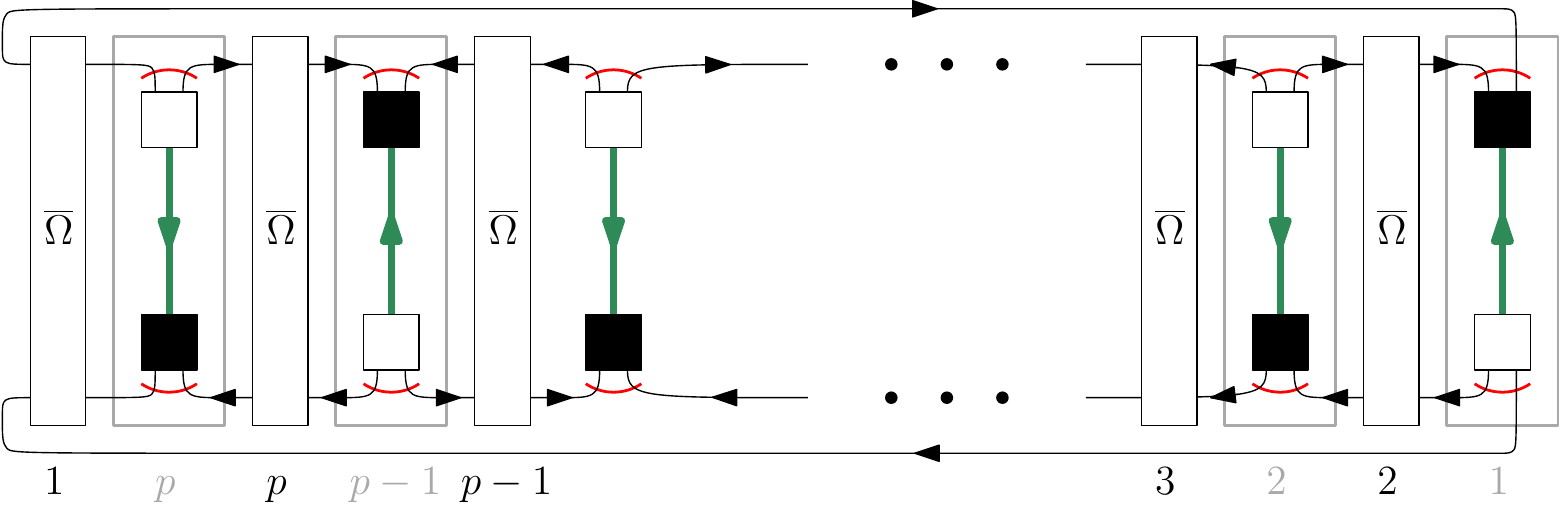}}
\end{equation}
To compute $\E(\Tr(Q^p))$ we proceed as before, with each Wick pairing identifying black box $i$ with white box $\alpha(i)$ and leading to a vertex labeled $i$. This vertex represents a $P_{s}$ operator. However, the resulting oriented graph also contains $\overline \Omega$-boxes that represents $P_{\overline \Omega}$.  This leads to a slightly more complicated connection pattern. We show on Fig. \ref{fig:local-structure-TN} the local structure around a vertex $i$. \\

\begin{figure}
    \centering
    \includegraphics[scale=0.7]{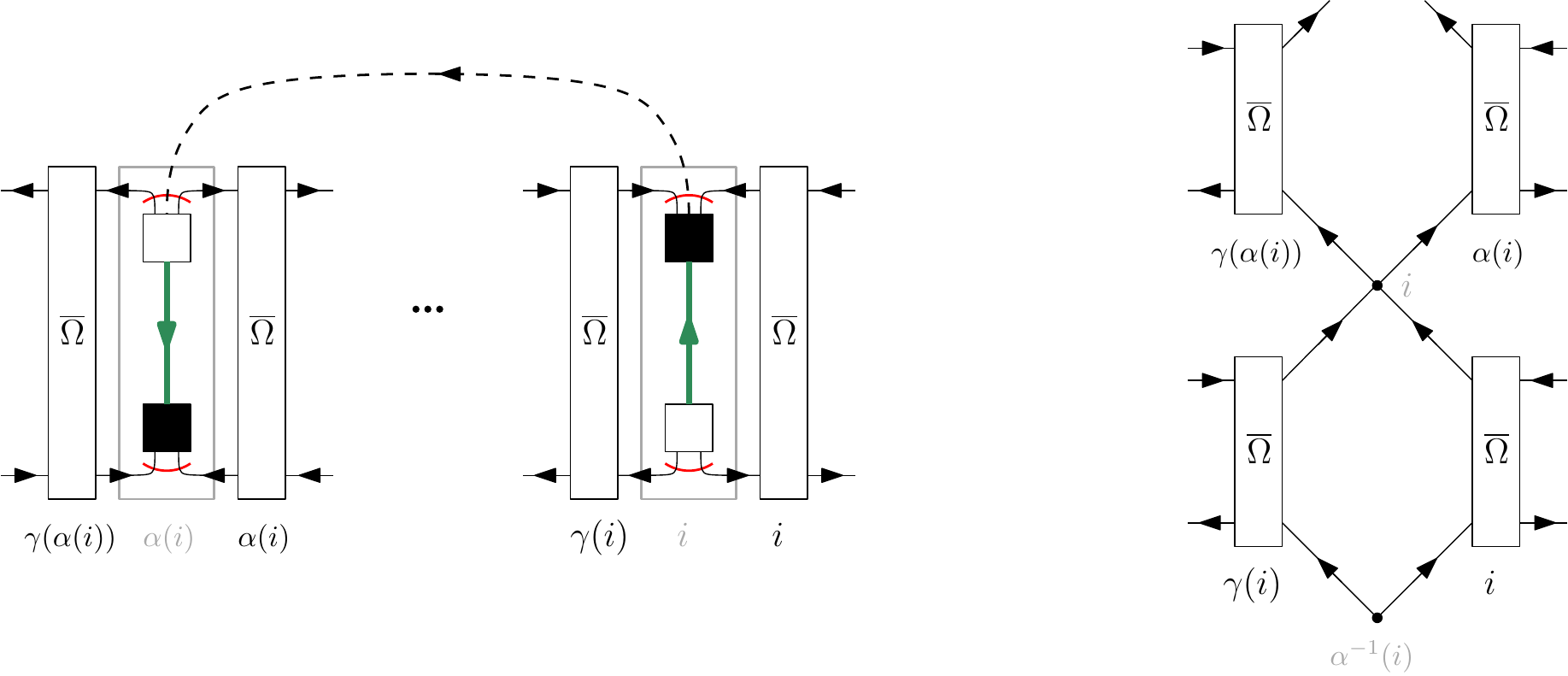}
    \caption{Local structure of a tensor network given a Wick pairing $\alpha$. Left: the Wick pairing on the diagram representing $\E(\Tr(Q^p))$. Right: the corresponding local structure of the tensor network of $P_{s}$ and $P_{\overline \Omega}$ operators.}
    \label{fig:local-structure-TN}
\end{figure}

\noindent The only change is that now vertices are only connected to $\overline \Omega$-boxes. Locally it is simple to check that vertex $i$ is connected to $\overline \Omega$-boxes $i$, $\gamma(i)$ \textit{via} ingoing edges, while it is connected to $\overline \Omega$-boxes $\alpha(i)$, $\gamma(\alpha(i))$ \textit{via} outgoing edges. This is due to the fact that the white box $\alpha(i)$ (of which adjacent edges are outgoing) is connected to $\overline \Omega$-boxes $\alpha(i)$, $\gamma(\alpha(i))$ while the black box $\alpha(i)$ (of which adjacent edges are ingoing) is connected to $\overline \Omega$-boxes $i$ and $\gamma(i)$. In the rest of this paper we denote by $T_{\alpha}$ the tensor network (and its evaluation) corresponding to a pairing $\alpha$ constructed in the way described above.\\

Using this tensor network construction, we have
\begin{equation}
    \E \Tr(Q^p)=\sum_{\alpha \in S_p} M^{\#\alpha} T_{\alpha}.
\end{equation}
Note that we have the relation
\begin{equation}\label{eq:TN-ev-vs-Graph-ev}
    T_{\alpha}=\frac1{2^p}\left(J(G_{\gamma\alpha,\gamma^{-1}\alpha}; N) + \sum_{\substack{f\in \{0,1\}^p \\ \exists i, f(i)=1}} \left(-1\right)^{|P_f|} N^{- \lvert P_f\rvert } J(G_{\Gamma_{a,P_f}\alpha,\Gamma_{b,P_f}\alpha };N)\right).
\end{equation}
In the next paragraph, we prove reduction relations for the evaluation of such tensor networks at fixed points, and simple transpositions of $\alpha$. This will allow us to consider permutations that do not contain fixed points and only have cycles of length $2$.\\

\noindent{\bf Reduction property for tensor network evaluation}\\

\noindent Le $c_q$ be a cycle of a permutation $\sigma = c_1 c_2 \ldots c_{\#\sigma}\in S_p$. We denote by $\sigma \div c_q:=c_1c_2\ldots \hat{c_q}\ldots c_{\# \sigma}$ the permutation acting on the set $\{1,\ldots,p\}\setminus \{i \, : \, i\in c_q\}$. We have the following reduction property for tensor network $T_{\alpha}$ such that $\alpha$ have the corresponding feature
\begin{proposition}\label{prop:TN-red-moves}
   Let $\alpha \in S_p$. Then,
   \begin{itemize}
       \item if $\alpha$ as a fixed point, that we denote $i$, then 
       \begin{equation}\label{eq:TN-to-reduced-TN}
           T_{\alpha}=\frac12 T_{\alpha\div (i)},
       \end{equation}
       where $\tilde{\alpha}$ is the permutation satisfying $\alpha=\tilde{\alpha}(i)$ in cycle notation. In particular, if $\alpha$ has several fixed points, consider $\alpha_R$ the reduced permutation acting on $\{1,\ldots, p\}\backslash \operatorname{Stab}(\alpha)$, then one has 
       \begin{equation}
           T_{\alpha}=\left( \frac12 \right)^{|\operatorname{Stab}(\alpha)|}T_{\alpha_R}
       \end{equation}
       \item if $\alpha$ contains a cycle of the form $(i,i+1)$, then 
       \begin{equation}
        T_{\alpha}=\Lambda(N)T_{\alpha \div (i,i+1)},
       \end{equation}
       where $\Lambda(N)=\frac14(N^2+2(N-\frac1{N}))$.
   \end{itemize}
   Note that the special cases where $\alpha=(1)$ and $\alpha=(12)$ are dealt with by setting $T_{\alpha\div(1)}=\Tr(P_{\overline \Omega})=N^2-1$ for the first case and $T_{\alpha\div (12)}=\Tr(P_{\overline \Omega})=N^2-1$ for the second case. For convenience, we denote $T_{\emptyset}:=\Tr(P_{\overline \Omega})$.
\end{proposition}
\begin{proof}
We proceed to the proof by diagrammatic manipulation of the tensor network. Fixed points $i$ of $\alpha$ can be treated locally by realizing that the corresponding tensor network always have the same structure around vertex $i$ (see the left hand side of equation \eqref{eq:fixed-point-reduction})
\begin{align}\label{eq:fixed-point-reduction}
    T_{\alpha}\ = \ \raisebox{-20mm}{\includegraphics[scale=0.42]{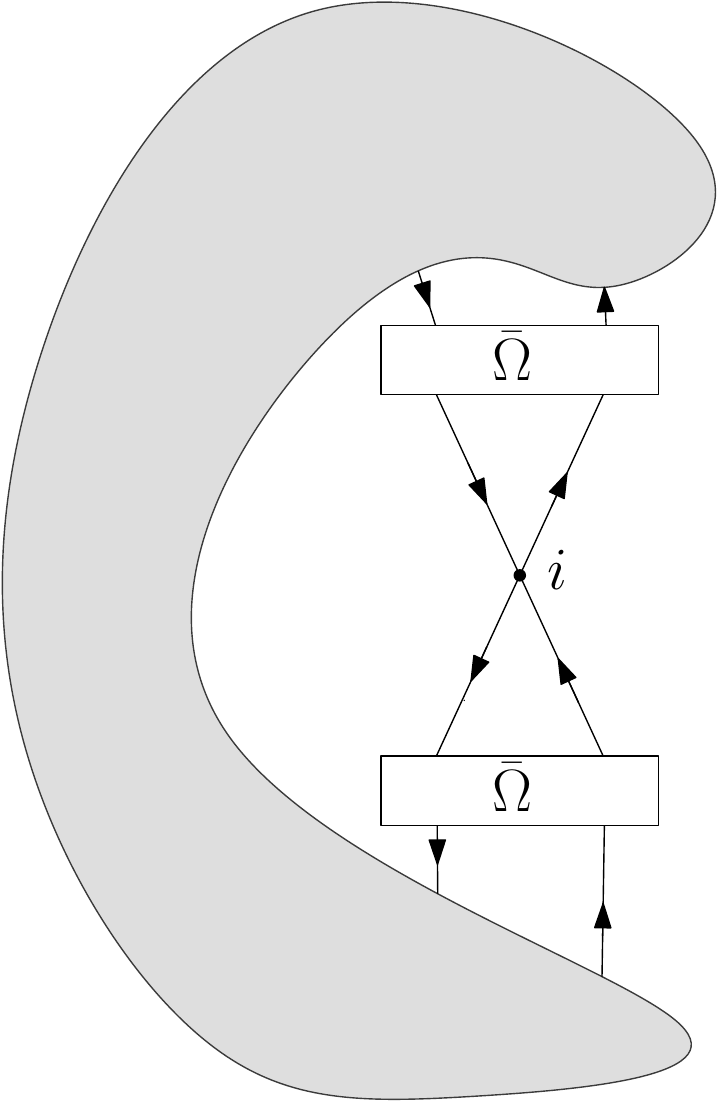}}\ &= \ \frac12 \ \raisebox{-20mm}{\includegraphics[scale=0.42]{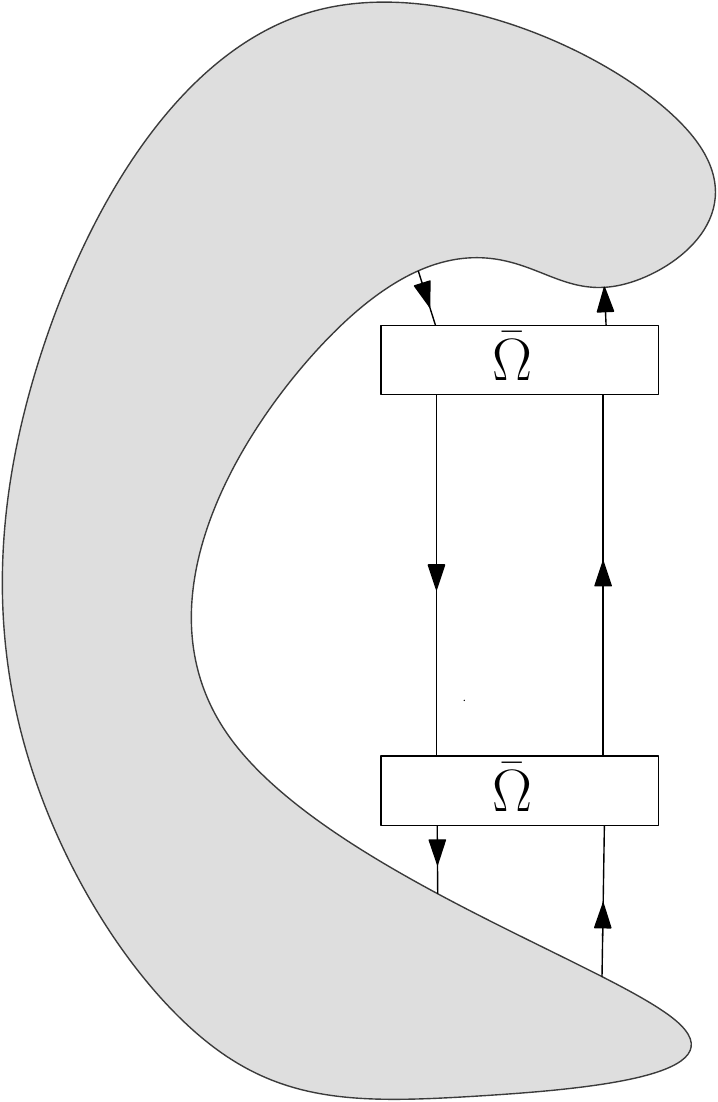}} + \frac12 \ \raisebox{-20mm}{\includegraphics[scale=0.42]{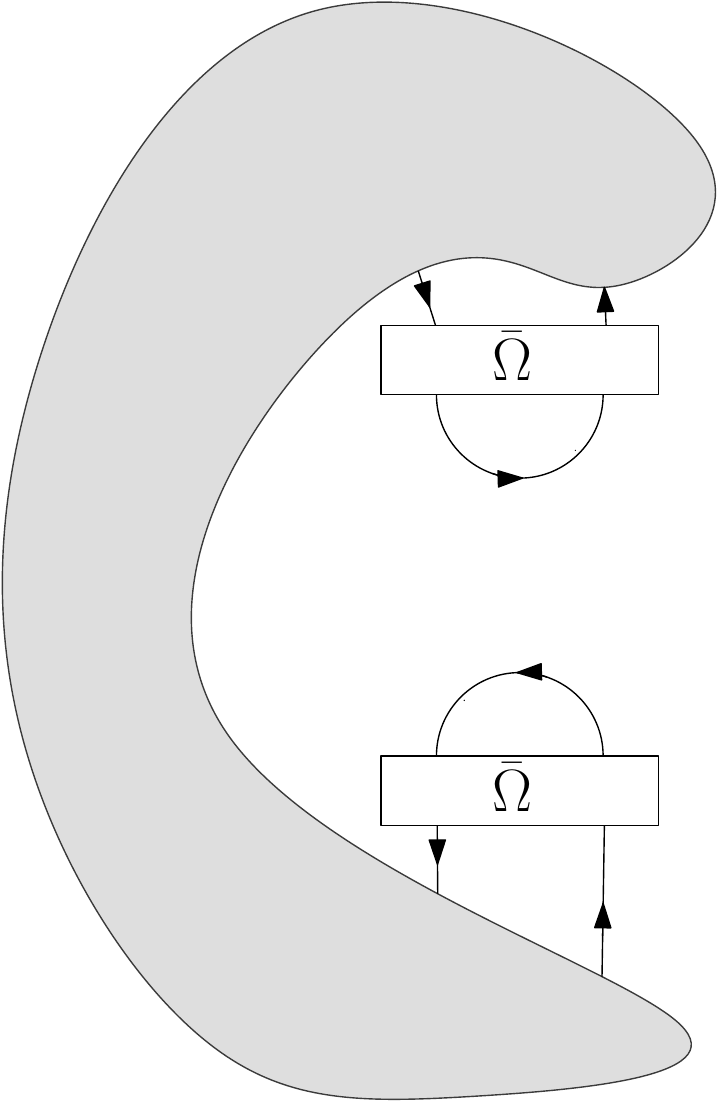}} \ = \frac12 \ \raisebox{-20mm}{\includegraphics[scale=0.4]{Fixed-points-red-rhs-1.pdf}}
\end{align}
where we used the fact that 
\begin{equation}\label{eq:PcompOmega-graphic}
    \raisebox{-8mm}{\includegraphics[scale=0.7]{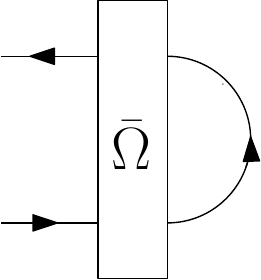}}\ = \ \sqrt{N}\Pcomp \lvert \Omega \rangle=0.
\end{equation}
Together with the fact that $\Pcomp^2=\Pcomp$ his shows that 
\begin{equation}
           T_{\alpha}=\frac12 T_{\alpha\div (i)}.
\end{equation}
The second statement can be obtained by looking at what happens locally on the vertices $i$ and $i+1$ when $\alpha$ has a cycle of the form $(i,i+1)$. One has the following diagrammatic relation
\begin{equation}\label{eq:i-i+1-expansion}
    T_{\alpha} \ = \ \raisebox{-20mm}{\includegraphics[scale=0.42]{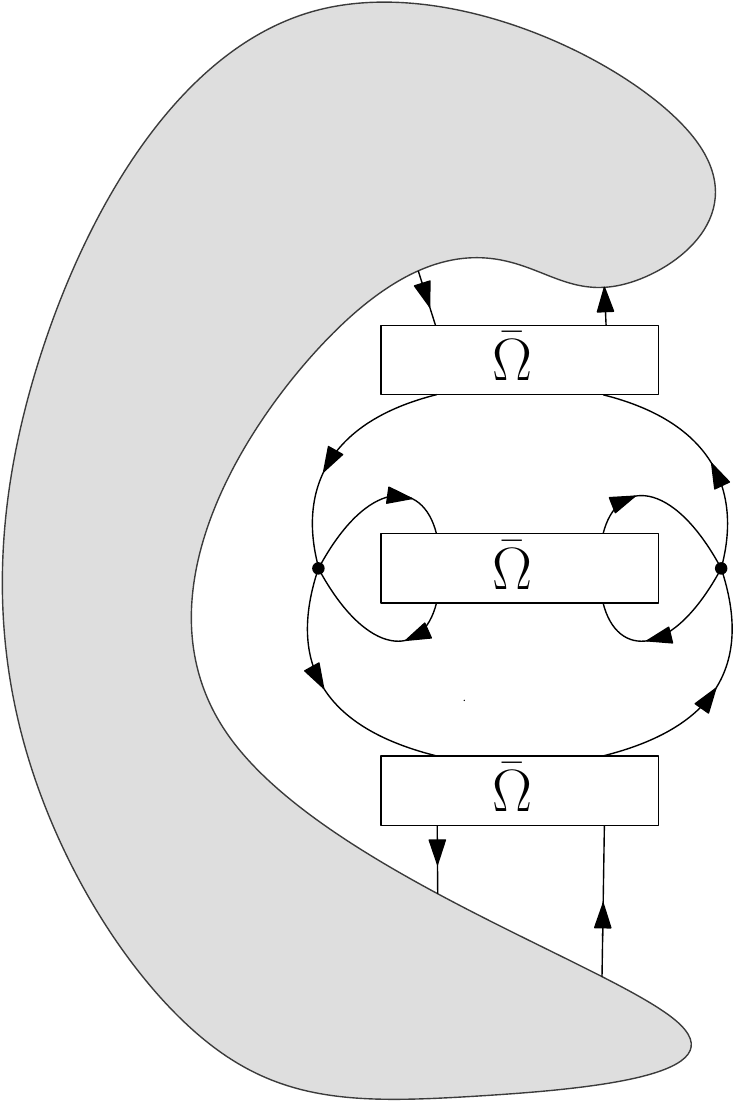}}\ = \ \raisebox{-20mm}{\includegraphics[scale=0.42]{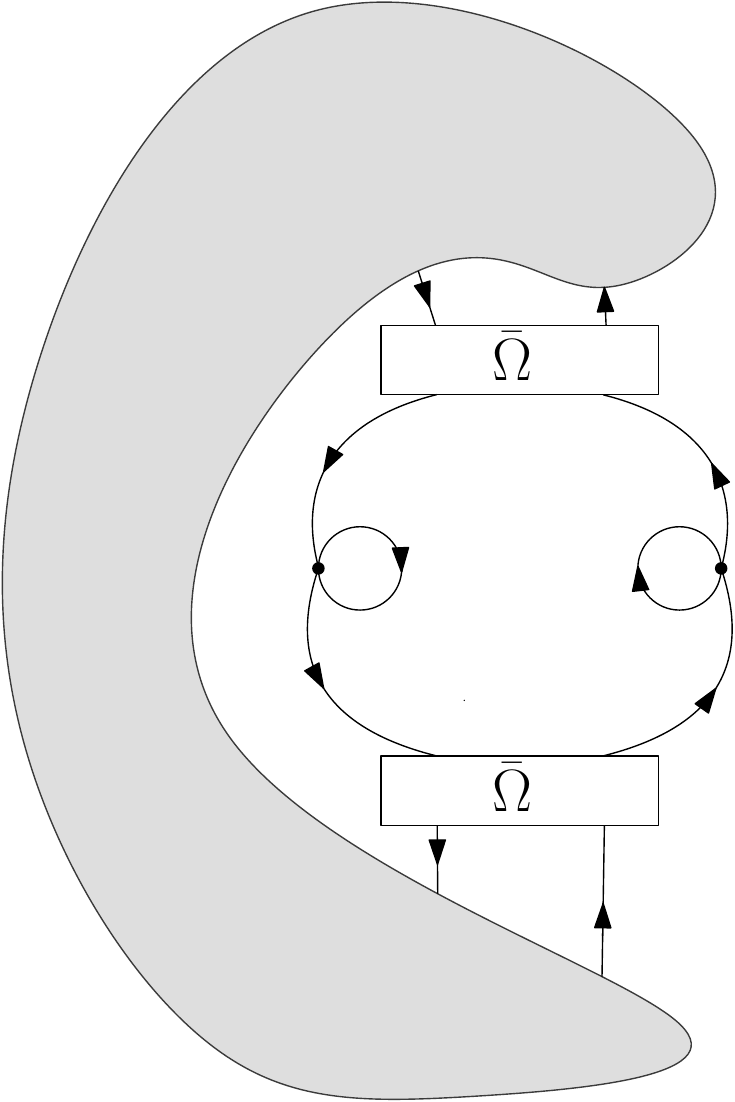}} -\frac1{N} \ \raisebox{-20mm}{\includegraphics[scale=0.42]{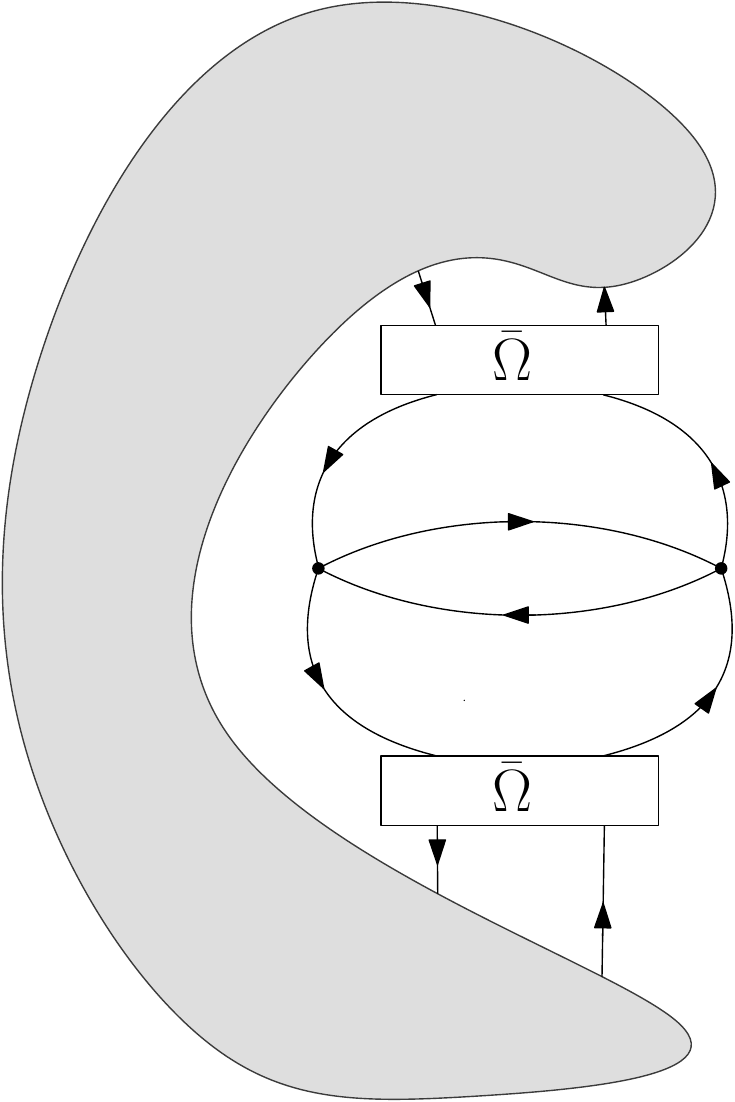}}
\end{equation}
obtained by using the definition of $\Pcomp$. The leftmost term of the right hand side of equation \eqref{eq:i-i+1-expansion} leads after further expanding the two black vertices 
\begin{equation}
    \raisebox{-20mm}{\includegraphics[scale=0.42]{i-i+1-red-rhs-1a.pdf}} = \frac14 (N+1)^2 \ \raisebox{-20mm}{\includegraphics[scale=0.4]{Fixed-points-red-rhs-1.pdf}}
\end{equation}
while the rightmost term of the right hand side of equation \eqref{eq:i-i+1-expansion} can also be extended further. Together with using again the relation \eqref{eq:PcompOmega-graphic} we have
\begin{equation}
    \frac1{N} \ \raisebox{-20mm}{\includegraphics[scale=0.42]{i-i+1-red-rhs-1b.pdf}} \ = \ \frac14 \left(\frac{N+1}{N}+\frac1{N}\right) \ \raisebox{-20mm}{\includegraphics[scale=0.4]{Fixed-points-red-rhs-1.pdf}}.
\end{equation}
Putting everything together we obtain 
\begin{equation}
    T_{\alpha}= \frac14 \left((N+1)^2-\frac{N+1}{N}-\frac1{N}\right) \ \raisebox{-20mm}{\includegraphics[scale=0.4]{Fixed-points-red-rhs-1.pdf}}
\end{equation}
hence $\Lambda(N)=\frac14(N^2+2(N-\frac1{N}))$.
\end{proof}
Thanks to this proposition we can consider permutations without fixed points by just remembering that for every $\alpha\in S_p$, there is a reduced permutation $\alpha_R$ without fixed points such that equation \eqref{eq:TN-to-reduced-TN} is satisfied. In particular it is important to note that by getting rid of the fixed points, we did not produce any factor of $N$. In terms of the right hand side of \eqref{eq:TN-ev-vs-Graph-ev}, this is due to cancellations between terms in the sum over words $f$. The tensor networks formulation allows to formulate these cancellations compactly. One particular case is if $\alpha\in S_p$ is the identity permutation. Then we have 
\begin{equation}T_{\alpha=id}=\frac1{2^p}(N^2-1).\end{equation}

Another particular case is when $\alpha\in \textrm{NC}_{2}(p)$, in that case it is always possible to find a cycle of the form $(i,i+1)$ and reducing it leads to another permutation $\alpha'$ which is in $\textrm{NC}_2(p-2)$. Using this property we can recursively reduce all the cycles of such a permutation and we obtain
\begin{equation}
    T_{\alpha\in \textrm{NC}_2(p)}=\Lambda(N)^{p/2}(N^2-1)=\frac1{2^p}N^{p+2}+O(N^p).
\end{equation}
\section{Proof of the main result}\label{sec:proof-main-result}
We now have all the tools we need to come to the proof of the main theorem \ref{thm:main}. This proof follows the steps:
\begin{enumerate}
    \item Determining the limiting spectrum of $Q$. This is Theorem \ref{thm:eigs-Q}. It is proved by nested applications of Proposition \ref{prop:TN-red-moves}, Lemma \ref{lem:maximal-scaling-cycle-type} and \ref{lem:crossing-length2}.
    \item Showing a  localization result for the overlap $\left\langle \Omega \left\lvert \frac{W^{\Gamma}}{N^3}\right\rvert \Omega\right \rangle$ in Proposition \ref{prop:moments-overlap}.
    \item Proving the main theorem by using Proposition  \ref{prop:moments-overlap} and Theorem \ref{thm:moments-WGamma-3-asympt} to control the largest and smallest eigenvalues of $W^\Gamma$ and using Cauchy's interlacing theorem on $\frac{W^\Gamma}{N^3}$ and $Q$ to obtain the convergence of the empirical distribution of the $N^2-1$ smallest eigenvalues of $W^{\Gamma}$.
\end{enumerate}

\noindent{\bf Limiting spectrum of $Q$.} We tackle the first step of our three steps plan of proof. 
Assume that the permutation $\alpha$ do not contain fixed points. By virtue of Lemma \ref{lem:maximal-scaling-cycle-type}, we know that $e\le 2$ and $e'\le 0$ for such $\alpha$. Moreover, these inequalities can be saturated if $\alpha$ has cycles of length $2$. This implies that there exists a positive constant $R$, independent of $N$,
\begin{equation}
    \lvert T_{\alpha} \rvert \le R N^{2p+2}
\end{equation}
which can be attained if $\alpha$ has only length $2$ cycles. This remark allows us to prove the following theorem
\begin{theorem}\label{thm:eigs-Q}
The limiting moments of the matrix $Q$ are given by 
\begin{equation}
\frac1{N^{2p+2}}\E\left(\Tr(Q^p)\right)=\frac1{2^p}\sum_{\alpha\in \textrm{NC}_{1,2}(p)} c^{\#\alpha}. 
\end{equation}
This completely determines the limiting law of eigenvalues of $Q$ as being a shifted semi-circular law of mean $c/2$ and variance $c/4$. In particular, the limiting support is $[-\sqrt{c}+c/2,\sqrt{c}+c/2]$ which is contained in the positive real line if and only if $c\ge4$.
\end{theorem}
\begin{proof}
We start with the relation
\begin{equation}
    \E\left( \Tr(Q^p)\right)=\sum_{\alpha\in S_p}N^{2\#\alpha} c^{\#\alpha} T_{\alpha}.
\end{equation}
Thanks to Proposition \ref{prop:TN-red-moves} we can reduce to permutations $\alpha'$ that do not contain fixed points at the cost of a factor $\frac12$, that is
\begin{equation}
    \E\left( \Tr(Q^p)\right)=\sum_{q=0}^p\frac1{2^q}\binom{p}{q}\sum_{\alpha'\in S_{p-q}}N^{2q +2\#\alpha'} c^{q+ \#\alpha'} T_{\alpha'}.
\end{equation}
According to Lemma \ref{lem:maximal-scaling-cycle-type}, we have 
\begin{equation}
    \sum_{q=0}^p\frac1{2^q}\binom{p}{q}\sum_{\alpha'\in S_{p-q}}N^{2q +2\#\alpha'} c^{q+ \#\alpha'} T_{\alpha'}=\sum_{q=0}^p\frac1{2^q}\binom{p}{q}\sum_{\substack{\alpha'\in S_{p-q}\\ \#_2\alpha'=\frac{p-q}{2}}}N^{2q +2\#\alpha'} c^{q+ \#\alpha'} T_{\alpha'} +O(N^{2p+1})
\end{equation}
as indeed all the cycles of $\alpha'$ must be of length $2$ since $\alpha'$ is not allowed to have fixed points. Making use of Lemma \ref{lem:crossing-length2} and the relation of equation \eqref{eq:TN-ev-vs-Graph-ev}, we can further reduce the sum 
\begin{equation}
    \sum_{\substack{\alpha'\in S_{p-q}\\ \#_2\alpha'=\frac{p-q}{2}}}N^{2q +2\#\alpha'} c^{q+ \#\alpha'} T_{\alpha'}=\sum_{\alpha'\in \textrm{NC}_2(p-q)} N^{2q +2\#\alpha'} c^{q+ \frac{p-q}{2}} T_{\alpha'} + O(N^{2p+1})
\end{equation}
because if $\alpha$ only have cycles of length $2$ and $\alpha\notin \textrm{NC}_2(p-q)$ then $\alpha$ can be reduced to a non trivial permutation $\tilde \alpha$ that do not contain cycles of the form $(i,i+1)$ by using Proposition \ref{prop:TN-red-moves} and recursively reducing cycles of this type. Then, $\tilde \alpha$ must have two cycles of the form $(a,i+1), (i,b)$ with $a<i<i+1<b$. This allows us to use Lemma \ref{lem:crossing-length2} to bound the contribution of $T_{\tilde \alpha}$. Finally, using the consequence of Proposition \ref{prop:TN-red-moves} for permutations in $\textrm{NC}_2(k)$, we have that 
\begin{align}
    \sum_{\alpha'\in \textrm{NC}_2(p-q)} N^{2q +2\#\alpha'} c^{q+ \frac{p-q}{2}} T_{\alpha'} &=\textrm{Cat}_{\frac{p-q}{2}}N^{2q +2\#\alpha'} c^{q+\frac{p-q}{2}}\Lambda(N)^{\frac{p-q}{2}}T_{\emptyset}\\
    &=\frac1{2^{p-q}}c^{q+\frac{p-q}{2}}\textrm{Cat}_{\frac{p-q}{2}}N^{2p}(N^2-1),
\end{align}
provided $p-q$ is even. Otherwise the sum is zero because there are no terms in the sum. Putting these nested sums together, we have that 
\begin{align}
    \E(\Tr(Q^p))=\frac1{2^p}\sum_{q=0}^p\binom{p}{q}\mathbf{1}_{p-q=0 \textrm{ mod }2}c^{q+\frac{p-q}{2}}\textrm{Cat}_{\frac{p-q}{2}}N^{2p+2}+O(N^{2p+1}).
\end{align}
This sum rewrites as 
\begin{equation}
    \E(\Tr(Q^p))=\frac{N^{2p+2}}{2^p}\sum_{\alpha\in \textrm{NC}_{1,2}(p)}c^{\#\alpha}+O(N^{2p+1}).
\end{equation}
\end{proof}
\noindent{\bf Localization of the overlap.} We now come to the second step of our three steps plan. We state and prove the localization result for the overlap below.

\begin{proposition}\label{prop:moments-overlap}
For all $p \geq 1$, 
$$\E  \left\langle \Omega \left\vert \frac{W^{\Gamma}}{N^3} \right\vert \Omega \right\rangle^p = N^{-4p} \sum_{\alpha \in S_p} M^{\#\alpha} N[2]^{\#\alpha}.$$
In particular, we have 
$$\E \left\langle \Omega \left\vert \frac{W^{\Gamma}}{N^3} \right\vert \Omega \right\rangle = \frac c 2, \qquad\qquad \operatorname{Var} \left\langle \Omega \left\vert \frac{W^{\Gamma}}{N^3} \right\vert \Omega \right\rangle =  \frac{M N[2]}{N^8} \sim \frac c 2 N^{-4},$$
and
$$\lim_{N \to \infty} \E \left\langle \Omega \left\vert \frac{W^{\Gamma}}{N^3} \right\vert \Omega \right\rangle^p = \left(\frac c 2\right)^p.$$
\end{proposition}
\begin{proof}
To show the first statement, use the Wick formula to write
$$\E  \left\langle \Omega \left\vert \frac{W^{\Gamma}}{N^3} \right\vert \Omega \right\rangle^p = N^{-4p} \sum_{\alpha \in S_p} M^{\#\alpha} \Tr_{\alpha, \alpha}(P_s).$$
Indeed, we have, graphically, 
\begin{center}
    \includegraphics{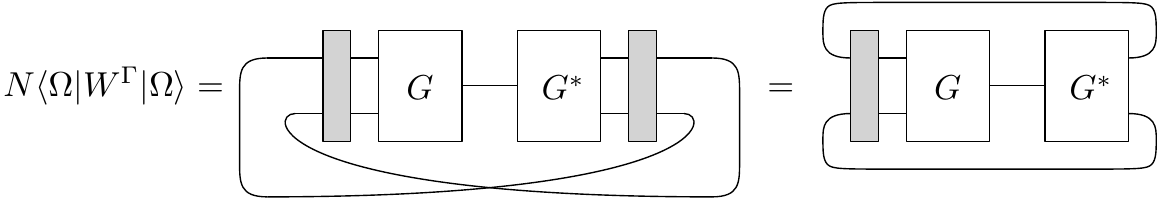}
\end{center}
Use now the fact that that $P_s$ is a projection, and thus
$$\Tr_{\alpha, \alpha}(P_s) = (\Tr P_s)^{\#\alpha} = N[2]^{\#\alpha}.$$
\end{proof}

\noindent {\bf Interlacing of eigenvalues and asymptotic of $W^{\Gamma}$.} We are now ready to prove the main result of the paper. As a reminder, let us start by stating a classical result, namely Cauchy's interlacing theorem
\begin{theorem}[Cauchy's interlacing theorem]\label{thm:cauchy-interlacing}
Let $A$ be a Hermitian operator on a Hilbert space $\mathcal{H}$, and let $B$ be its compression to a codimension $k$ subspace $\mathcal{N}$. Then for $j=1,2,\ldots, n-k$
\begin{equation}
 \lambda_j(A)\ge \lambda_{j}(B)\ge \lambda_{j+k}(A).
\end{equation}
\end{theorem}
We use this theorem to recover the behaviour of the limiting eigenvalues of $W^{\Gamma}$ from the limiting spectrum of $Q$ and the large eigenvalue of $W^{\Gamma}$. We now come to the proof of theorem \ref{thm:main}.\\

\begin{proof}[Proof of the Theorem \ref{thm:main}]
Let us first show that the largest eigenvalue of the matrix $W^\Gamma$ behaves, asymptotically as $N \to \infty$, as $\frac c 2 N^3$. To this end, we shall compare it with two quantities: the overlap $\langle \Omega | W^\Gamma | \Omega \rangle$ and wih the moments of the random matrix $W^\Gamma$.

On the one hand, we have (we denote $\lambda_{\max} = \lambda_1$ for clarity), for all $p \geq 2$, 
$$\Tr\left( \frac{W^\Gamma}{N^3}\right)^{2p}\ge \lambda_{\max}^{2p}\left( \frac{W^\Gamma}{N^3}\right)\ge \left\langle \Omega\left\lvert \frac{W^\Gamma}{N^3} \right\rvert \Omega\right\rangle^{2p}.$$
Using Theorem \ref{thm:moments-WGamma-3-asympt} and Proposition \ref{prop:moments-overlap}, we conclude that
$$\lim_{N \to \infty} \lambda_{\max}^{2p}\left( \frac{W^\Gamma}{N^3}\right) = \left(\frac c 2 \right)^{2p}.$$
By Markov's inequality, we have thus 
$$\lim_{N \to \infty}\lambda_{\max}^{2}\left( \frac{W^\Gamma}{N^3}\right) = \left(\frac c 2 \right)^{2}$$
in probability. But since $\Tr W^\Gamma = \Tr W \geq 0$, $\lambda_{\max}(W^\Gamma) \geq 0$ almost surely, so we also have, in probability,
$$\lim_{N \to \infty}\lambda_{\max}\left( \frac{W^\Gamma}{N^3}\right) = \frac c 2,$$
establishing the first claim.

In order to prove the second part of the statement, concerning the lower $N^2-1$ eigenvalues of $W^\Gamma$, we shall make use of Theorems \ref{thm:eigs-Q} and \ref{thm:cauchy-interlacing}. Indeed, if we denote by $\mu_1 \geq \cdots \geq \mu_{N^2-1}$ the (non-zero) eigenvalues of the matrix $N^{-2}Q$ and by $\lambda_1 \geq \lambda_2 \geq \cdots \geq \lambda_{N^2-1} \geq \lambda_{N^2}$ those of $N^{-2}W^\Gamma$, by Cauchy's interlacing theorem, we have 
$$\lambda_1 \geq \mu_1 \geq \lambda_2 \geq \mu_2 \geq \cdots \geq \lambda_{N^2-1} \geq \mu_{N^2-1} \geq \lambda_{N^2}.$$
On the other hand, we know by Theorem \ref{thm:eigs-Q} that the empirical distribution of the $\mu$ eigenvalues converges to the desired shifted semicircle distribution. Hence, in order to conclude, we need to control the smallest eigenvalue $\lambda_{N^2}$ and show that it does not contribute asymptotically to moments. We have 
\begin{align*} \E \lambda_{\min}^4\left(\frac{W^\Gamma}{N^3}\right)& \leq \E \Tr\left(\frac{W^\Gamma}{N^3}\right)^4 - \E \lambda_{\max}^4\left(\frac{W^\Gamma}{N^3}\right) \\
& \leq \E \Tr\left(\frac{W^\Gamma}{N^3}\right)^4 - \E \left\langle \Omega \left\vert \frac{W^{\Gamma}}{N^3} \right\vert \Omega \right\rangle^4 \leq o(1) + \left(\frac c 2 \right)^4 - \left(\frac c 2 \right)^4 = o(1).
\end{align*}
Hence, $\E|\lambda_{\min}(W^\Gamma)| = o(N^{3})$, which implies that $\E |\lambda_{N^2}| = o(N)$, completing the proof. 
\end{proof}

\section{Conclusion}\label{sec:conclusion}
In this paper we have introduced the ensembles of bosonic and fermionic density matrices, which are random quantum states supported on the symmetric, resp.~anti-symmetric subspaces of a tensor product space. We have then studied the entanglement of typical states from these ensembles, focusing on the partial transposition criterion. 

In the fermionic case, we have found that the partial transposition of such a random density matrix has a large negative eigenvalue, hence a typical fermionic bipartite density matrix is entangled. The bosonic case is more delicate, since the symmetry condition imposes a large positive eigenvalue for the partial transposition, so a finer analysis of the spectrum is required. 

We use the moment method to completely describe the spectrum of the partial transposition of a symmetric random density matrix. We find that the bulk of the spectrum follows a shifted semicircular distribution at scale $N^2$, similarly to the non-symmetric case discussed in \cite{aubrun2012partial}. We also find a large outlier positive eigenvalue, on the scale $N^3$, which is a signature of the symmetry. Our method relies on establishing a connection between the moments of symmetric random matrices and the circuit counting graph polynomial. This allows us to characterize the asymptotic ratio between the system size and the environment size for which the PPT criterion is strong enough to certify entanglement. 

The current paper is a first step in the study of the entanglement properties of the bosonic and fermionic ensembles of quantum states introduced here, several questions being left for future work. In the bipartite case (which is the main focus of this paper), the question of the convergence of the smallest eigenvalue of the partial transposition towards the left edge of the spectrum is left open. Probably, this would require the detailed analysis of large moments of the corresponding random matrix, similarly to the method used in \cite{aubrun2012partial}; such an analysis is hindered in the current situation by the presence of the outlier eigenvalue on a larger scale. 

We have not addressed the multipartite case, $r \geq 3$, mainly because the tools used in this work (connection to graph polynomials) are not efficient enough in the general case. However, first investigations of this general case suggest that an interesting phenomenon is at play. Indeed, we expect the spectrum of the partial transpose to split at several scales, such that the empirical eigenvalues distribution has several ($r-1$) connected dense parts with each part appearing at a specific scale. We expect that the rigorous study of this general case will make heavy use of representation theory. 

Another interesting direction for future study is the de Finetti theorem for the bosonic ensemble. This would require analyzing the separability property of the partial trace of a symmetric random density matrix in the regime where the parameter $r$ is large, and the Hilbert space dimension $N$ is fixed. This asymptotic regime usually requires completely different techniques in random matrix theory. 

In conclusion, we have set up and studied basic entanglement properties of symmetric (and anti-symmetric) random density matrices. Several questions regarding these ensembles are left open, and we hope that our work will generate interest in these ensembles of random matrices / tensors, and their connection to combinatorics and graph polynomials.

\bigskip

\paragraph*{\bf Acknowledgments.}
S.~Dartois has been partially supported by the Australian Research Council grant DP170102028, the French ANR project ANR-18-CE47-0010 (\href{https://www.irif.fr/~magniez/qudata/}{QUDATA}). I.~Nechita has been supported by the ANR project \href{https://esquisses.math.cnrs.fr/}{ESQuisses}, grant number ANR-20-CE47-0014-01. Both S.~Dartois and I.~Nechita were supported on this project by the ANR project ``Investissements d'avenir'' ANR-11-LABX-0040. A. Tanasa has been partially supported by the ANR-20-CE48-0018 ``3DMaps" grant.

\bibliographystyle{alpha}
\bibliography{tensor-QIT}
\end{document}